\documentclass[leqno,fleqn,submission,copyright,creativecommons]{eptcs}
\providecommand{\event}{DSL 2011} 

\title{A Domain-Specific Language for Incremental and Modular 
       Design of Large-Scale Verifiably-Safe Flow Networks (Preliminary Report)
       }
\author{Azer Bestavros
\institute{Boston University}
\email{best@bu.edu}
\and
Assaf Kfoury
\institute{Boston University}
\email{kfoury@bu.edu}
 }
\def\titlerunning{DSL for Design of Verifiably-Safe Flow Networks}
\def\authorrunning{A. Bestavros \& A. Kfoury}

\usepackage{amsthm} 
\usepackage{url}
\usepackage{times}
\usepackage{bm}   
\usepackage{stmaryrd} 
\usepackage{latexsym}
\usepackage{amsmath}
\usepackage{amssymb}
\usepackage{amsfonts}
\usepackage{MnSymbol}  
\usepackage{mathrsfs}  
\usepackage{graphicx} 
\usepackage{wrapfig}
\usepackage[all]{xy}
\usepackage{fancybox}
\usepackage[usenames]{color}

\newcommand{\Hide}[1]{}

\newtheorem{theorem}{Theorem}
\newtheorem{lemma}[theorem]{Lemma}
\newtheorem{corollary}[theorem]{Corollary}
\newtheorem{proposition}[theorem]{Proposition}

\newtheorem{property}[theorem]{Property}
\newtheorem{termxxx}[theorem]{Terminology}

\newtheorem{notxxx}[theorem]{Notation}

\newtheorem{convenxxx}[theorem]{Convention}

\newtheorem{exaxxx}[theorem]{Example}
\newenvironment{example}{\begin{exaxxx}\rm}{\hfill\QED\end{exaxxx}}
\newtheorem{exexxx}[theorem]{Exercise}

\newtheorem{remxxx}[theorem]{Remark}
\newenvironment{remark}{\begin{remxxx}\rm}{\hfill\QED\end{remxxx}}
\newtheorem{openxxx}[theorem]{Open Problem}

\newtheorem{conjxxx}[theorem]{Conjecture}

\newtheorem{defxxx}[theorem]{Definition}
\newenvironment{definition}[1]{\begin{defxxx}[\emph{#1}]\rm}{\hfill\QED\end{defxxx}}

\newenvironment{sketch}
{\smallskip\noindent\ignorespaces\textit{Proof Sketch.}}
{\hfill\QED\medskip}

\usepackage[margin=0pt,justification=centerlast,font=small,format=hang,labelfont=bf,up,textfont=rm,up]{caption}

\newcommand{\Angles}[1]{\langle #1 \rangle}
\newcommand{\Set}[1]{\{ #1 \}}
\newcommand{\SET}[1]{\bigl\{ #1 \bigr\}}
\newcommand{\Judgement}[3]{#1\,\vdash\, #2 : #3}
\newcommand{\inn}[1]{\textbf{in}(#1)} 
\newcommand{\inter}[1]{\textbf{\#}(#1)} 
\newcommand{\out}[1]{\textbf{out}(#1)} 
\newcommand{\Env}{\Gamma} 
\newcommand{\aaa}{\textbf{A}} 
\newcommand{\bbb}{\textbf{B}} 
\newcommand{\nn}{\textbf{N}} 

\newcommand{\A}{{\cal A}}
\newcommand{\B}{{\cal B}}
\newcommand{\C}{{\cal C}}

\newcommand{\F}{{\cal F}}

\newcommand{\M}{{\cal M}}
\newcommand{\N}{{\cal N}}
\newcommand{\PP}{{\cal P}}

\newcommand{\rename}[2]{{^{{#2}}}\!{#1}} 

\newcommand{\Conn}[3]{\textbf{\textsf{conn}}(#2,#3,#1)}

\newcommand{\ConnN}[2]{#1\,\text{\small $\oplus$}\,#2}
\newcommand{\ConnNN}[2]{\bigoplus\bigl(#1,#2\bigr)}

\newcommand{\ConnP}[2]{#1\;\bfmath{\|}\;#2}
\newcommand{\bigPlus}{{\;\rule[2.5pt]{8pt}{1pt}\;}{\kern-10pt}{\;\rule[-1pt]{1pt}{8pt}\ \;}}  
\newcommand{\Merge}[3]{\textbf{\textsf{Merge}}\bigl(#1,#2,#3\bigr)}
\newcommand{\Fork}[3]{\textbf{\textsf{Fork}}\bigl(#1,#2,#3\bigr)}
\newcommand{\Loop}[2]{\textbf{\textsf{bind}\,}(#2,#1)}

\newcommand{\Let}[3]%
    {\textbf{\textsf{let}}\ {#1}\,{#2}\ \textbf{\textsf{in}}\;{#3}\,}
\newcommand{\Choose}[3]%
    {\textbf{\textsf{choose}}\ {#1} {#2}\ \textbf{\textsf{in}}\;{#3}\;}
\newcommand{\Try}[3]%
    {\textbf{\textsf{try}}\ {#1} {#2}\ \textbf{\textsf{in}}\;{#3}\;}
\newcommand{\Letrec}[3]%
    {\textbf{\textsf{letrec}}\ {#1} {#2}\ \textbf{\textsf{in}}\;{#3}\;}
\newcommand{\Mix}[3]%
    {\textbf{\textsf{mix}}\ {#1} {#2}\ \textbf{\textsf{in}}\;{#3}\;}
\newcommand{\LET}[3]%
    {\textbf{\textsf{let}}^*\ {#1} {#2}\ \textbf{\textsf{in}}\;{#3}\;}
\newcommand{\Partial}[2]{\,{\subseteq}_{\text{1-1}}\; #1\, {\times}\, #2}

\newcommand{\ie}{\textit{i.e.}}
\newcommand{\eg}{\textit{e.g.}}
\newcommand\QED{$\square$} 

\newcommand{\proj}[3]{\textsf{Proj}_{#1}^{#2}\bigl(#3\bigr)}  
\newcommand{\nreals}{\mathbb{R}^+}
\newcommand{\reals}{\mathbb{R}}
\newcommand{\bfmath}[1]{{\text{\large $\boldsymbol{#1}$}}}  
\newcommand{\tight}[1]{\textsf{Tight}(#1)} 
 
\newcommand{\poly}[1]{\textsf{Poly}(#1)} 
\newcommand{\Poly}[1]{\textsf{Poly}^*(#1)} 
\newcommand{\size}[1]{|#1|}

\newcommand{\power}[1]{\mathscr{P}(#1)}
\newcommand{\inT}[1]{\bfmath{[}#1\bfmath{]}_{\text{in}}}
\newcommand{\outT}[1]{\bfmath{[}#1\bfmath{]}_{\text{out}}}
\newcommand{\rest}[2]{\bfmath{[}#1\bfmath{]}_{#2}}
\newcommand{\head}[1]{\textit{head}(#1)} 
\newcommand{\tail}[1]{\textit{tail}(#1)} 
\newcommand{\loopT}[2]{\textsf{bind}(#1,#2)}
\newcommand{\ConnPT}[2]{#1\;{\|}\;#2}
\newcommand{\andT}[2]{#1\bfmath{\land}#2}
\newcommand{\ioSem}[1]{\bm{\llangle} #1 \bm{\rrangle}} 
\newcommand{\fullSem}[1]{\bm{\llbracket} #1 \bm{\rrbracket}} %
\newcommand{\dimI}[1]{\textsf{dim}_{\text{in}}(#1)} %
\newcommand{\dimO}[1]{\textsf{dim}_{\text{out}}(#1)} %
\newcommand{\dimIO}[1]{\textsf{dim}(#1)} %
\newcommand{\half}[1]{\textsf{Half}(#1)}

\newcommand{\symhr}{\textsc{hr}}

\newcommand{\LE}[3]{#2 {\,<}^{#1}\; #3}
\newcommand{\LEQ}[3]{#2 {\,\leqslant}^{#1}\; #3}
\newcommand{\hr}[1]{\symhr(#1)}
\newcommand{\symau}{\textsc{au}}
\newcommand{\au}[1]{\symau(#1)}
\newcommand{\symmd}{\textsc{md}}
\newcommand{\md}[1]{\symmd(#1)}
\newcommand{\objA}[2]{{\Phi}_{#1}^{#2}} 
\newcommand{\objB}[2]{{\varphi}_{#1}^{#2}}
\newcommand{\ooo}{\alpha} 

\newcommand{\FullSem}[1]{\fullSem{\!\!\fullSem{\,#1\,}\!\!}} 
\newcommand{\CC}{\mathscr{C}}

\newcommand{\spacing}[2]{
  \renewcommand{\baselinestretch}{#2}
  \small\normalsize #1
  \setlength{\parskip}{0.2\baselineskip}
  \settowidth{\parindent}{xxxx}
  \setlength{\parindent}{#2\parindent}
  \setlength{\leftmargini}{\parindent}
  \setlength{\leftmarginii}{\parindent}
  \setlength{\leftmarginiii}{\parindent}
  \setlength{\footnotesep}{#2\footnotesep}
}

\begin{document}

\spacing{\normalsize}{0.95}

\maketitle
%
\begin{abstract}

\noindent
We define a \emph{domain-specific language} (DSL) to inductively assemble
\emph{flow networks} from \emph{small networks} or 
\emph{modules} to produce arbitrarily large ones,
with interchangeable functionally-equivalent parts.  Our small
networks or modules are ``small'' only as the building blocks in this
inductive definition (there is no limit on their size).
Associated with our DSL is a \emph{type theory}, a system of formal
annotations to express desirable properties of flow networks together
with rules that enforce them as \emph{invariants} across their
interfaces, \ie, the rules guarantee the properties are preserved as
we build larger networks from smaller ones.
A prerequisite for a type theory is a \emph{formal semantics}, \ie, a
rigorous definition of the entities that qualify as feasible flows
through the networks, possibly restricted to satisfy additional
efficiency or safety requirements. This can be carried out in one of
two ways, as a \emph{denotational} semantics or as an
\emph{operational} (or \emph{reduction}) semantics; we choose
the first in preference to the second, partly to avoid exponential-growth
rewriting in the operational approach.
We set up a typing system and prove its soundness for our DSL.

\end{abstract}

\section{Introduction and Motivation}
\label{sect:motivation}
\noindent{\bf Flow Networks.} 
  Most large-scale systems can be viewed as assemblies of subsystems,
  or gadgets, each of which produces, consumes, or regulates a flow of
  some sort. In a computer network, a discrete flow of messages
  (packets) is produced by servers ({\em e.g.}, streaming sources),
  regulated by network devices ({\em e.g.}, routers and shapers), and
  consumed by clients ({\em e.g.}, stream players). In a road network,
  the flow constitutes vehicles which enter and exit at edge
  exchanges, and which are regulated by speed limits on road segments,
  and by traffic lights at forks and intersections. In electric grids,
  a continuous flow of energy (electric current flow) is produced by
  power sources, regulated by transformers, transported by
  transmission lines, and consumed by power sinks. In a sensor
  network, a flow of measurements is produced by sensors, regulated by
  filters and repeaters, and consumed by sinks and aggregators. In a
  computing grid or cloud, a flow of resources ({\em e.g.}, CPU
  cycles) is produced by physical clusters of hosts, regulated by
  schedulers, resource managers, and hypervisors, and consumed by
  applications.

In each of the above systems, a ``network'' is assembled from smaller
  building blocks, which themselves could be smaller, inductively
  assembled networks or alternately, they could be individual {\em
  modules}. Thus, what we call {\em flow networks} are inductively
  defined as assemblies of {\em small networks} or {\em modules}. The
  operation of a flow network is characterized by a set of variables
  and a set of constraints thereof, reflecting {\em basic}, {\em
  assumed}, or {\em inferred} properties or rules governing how the
  network operates, and what constitutes safe operation. Basic rules
  (variables and constraints) are inherently defined, and are
  typically specified by a domain expert for individual
  modules. Assumed rules are speculatively specified for outsourced or
  yet-to-be fleshed out networks, which constitute {\em holes} in a
  larger network. Holes in a network specification allow the design or
  analysis of a system to proceed based only on promised functionality
  of missing modules or networks to be plugged in later.  Inferred
  rules are those that could be derived through repeated composition
  and analysis of networks. Such derivations may be exact, or may
  underscore conservative approximations ({\em e.g.}, upper or lower
  bounds on variables or expressions).

Basic or inferred rules -- underscoring constraints on the operation
  of a flow network -- could be the result of analysis using any one
  of a set of diverse theories or calculi. For instance, in a
  streaming network application, the size of a maximum burst of
  packets produced by a server over a window of time may be bounded
  using analysis that relies on real-time scheduling theory, whereas
  the maximum burst of packets emitted by a sequence of networking
  elements ({\em e.g.}, multicast routers and shapers) over a
  (possibly different) window of time may be bounded using analysis
  that relies on network calculus \cite{networkcalculus}. Clearly,
  when a larger flow network consisting of streaming servers as well
  as network elements -- not to mention holes -- is assembled, neither
  of these underlying calculi on its own could be used to perform the
  requisite network-wide analysis to derive the rules at the
  boundaries of the larger flow network. Rather, the properties at the
  boundaries of the constituent (smaller) networks of servers and
  networking elements constitute a domain-specific language (of
  maximum burst size over time, in this case), the semantics of which
  can be used to derive the rules at the boundaries of the larger flow
  network.

Several approaches to system design, modeling and analysis have been
  proposed in recent years, overlapping with our notion of flow
  networks. Apart from the differences in the technical details -- at
  the level of formalisms and mathematics that are brought to bear --
  our approach distinguishes itself from the others by incorporating
  from its inception three inter-related features/goals: (a) the
  ability to pursue system design and analysis without having to wait
  for missing (or broken) components/modules to be inserted (or
  replaced), (b) the ability to abstract away details through the
  retention of only the salient variables and constraints at network
  interfaces as we transition from smaller to larger networks, and (c)
  the ability to leverage diverse, unrelated theories to derive
  properties of modules and small networks, as long as such networks
  share a common formal language at their interfaces -- a formal
  Domain-Specific Language (DSL) that enables assembly and analysis
  that is agnostic to the underlying theory used to derive such
  properties.

\noindent{\bf Examples of DSL Use Cases.} 
  Before delving into the precise definitions and formal arguments of
  our DSL, we provide brief descriptions of how flow networks could be
  leveraged for two application domains -- namely resource allocation
  and arbitration subject to Service Level Agreements (SLAs) for
  \emph{video streaming} in a \emph{cloud computing} setting, and
  emerging safety-critical CPS and \emph{smart grid} applications.

The generality of our DSL is such that it can be applied to problems
  in settings that are not immediately apparent as flow network
  settings.  For example, consider a single, physical or virtual host
  (processor). One may view such a host $i$ as the source of a {\em
  supply flow} of compute cycles, offered in constant increments $c_i$
  every period $t_i$.  Similarily, a process or application $j$
  executing on such a host can be viewed as a {\em demand flow} of
  compute cycles, requested periodically with some characteristics --
  {\em e.g.}, subject to a maximum consumption of $w_j$ cycles per
  period $t_j$. In this setting, multiple supply flows ({\em e.g.} a
  set of processors in a multicore/cluster setting), each represented
  by an individual supply $(c_i,t_i)$ flow, can be regulated/managed
  using hypervisor system software to yield a flow network that
  exhibits a more elaborate pattern of compute cycles. For instance,
  the resulting flow may be specified as a single $(c_m,t_m)$ flow,
  where $c_m$ cycles are supplied over the Least Common Multiple (LCM)
  period $t_m$, or it may be sepcified as a set of $(c_k,t_k)$ flows,
  each of which operating at some discrete period $t_k$ drawn from the
  lattice of LCM periods defined by the individual $t_i$
  periods. Similarily, multiple demand flows ({\em e.g.} a set of
  services offered within a single virtual machine), each represented
  by an individual demand $(w_j,t_j)$ flow, can be multiplexed to
  yield more elaborate consumption patterns of the resulting
  workload. Finally, a supply flow may be matched up to a set of
  demand flows through the use of a scheduler.  Clearly, for a flow
  network of compute cycle producers, consumers, and schedulers to
  operate safely, specific constraints (rules) must be satisfied. For
  instance, matching up supply and demand flows adhere to a ``supply
  meets demand'' condition, or to some other SLA, such as ``periods of
  overload cannot exceed 100 msecs'' or ``no more than 5 missed
  periodic allocations in any 1-minute window of time''.

Not only is our DSL useful in modeling the supply of, demand for, and
  consumption (through a scheduler) of compute cycles, but also in a
  very similar manner they can be used readily to model the supply of,
  demand for, and consumption (through resource management protocols)
  of other computing resources such as network bandwidth, storage
  capacities, {\em etc.}

In the above setting, the flow networks describing the supply,
  demand, or scheduling of computing and networking resources can be
  made as small as desired to render their whole-system analysis
  tractable, or as large as desired to produce more precise
  system-wide typings. For instance, readers familiar with the vast
  literature on real-time scheduling ({\em e.g.}, 
  \cite{Liu1973Scheduling,Regehr2001HLS,Shin2003Periodic})
  will immediately recognize that most of the results in that
  literature can be viewed as deriving fairly tight bounds on specific
  processor schedulers such as EDF,
  RMS, Pinwheel, among others schedulers.  Similarily, readers
  familiar with QoS provisioning using network calculus, traffic
  envelopes, fluid network models will recognize that most of the
  results obtained through these models are applicable for specific
  protocols such as AIMD, weighted-fair queuing, among other
  schedulers ({\em e.g.}, 
  \cite{networkcalculus,255724,Thiran2001}).

Modeling and analysis of the supply of (and demand for) computing and
  networking resources is particularly valuable in the context of
  cloud and grid resource management ({\em e.g.},
  \cite{auyoung04resource,Buyya2000,Gomoluch2004,IshakianSwehaLondonoBestavros:nca10,662332}). In
  such a setting, a cloud operator may use a DSL to specify the
  topological configuration of computing and networking resources, the
  layer of system software used to virtualize these resources, as well
  as a particular mapping of client workloads to virtualized
  resources. Compiling such a DSL-sepecification is akin to verifying
  the safety of the system. Moreover, making changes to these DSL
  specifications enables the operator (or a mechanized agent thereof)
  to explore whether an alternative arrangement of resources or an
  alternative mapping of client workloads is more efficient
  \cite{IshakianBestavrosKfoury:rtsca10}.

As another example of the broad applicablity of our DSL, consider yet
  another application domain -- that of smart electric grids. In this
  domain, a module would be a grid ``cell'', such as a power plant, a
  residential or commercial building, a power transmission line, a
  transformer, or a power storage facility (batteries), {\em etc.}
  Each cell has a capacity to produce and consume power over time
  (energy flow). For example, a house with solar panels may be
  contributing a positive flow to the grid or a negative flow
  depending on the balance between solar panel supply and house
  demand. Operational or safety constraints on cells and
  interconnections of cells define relationships that may be the
  subject of exact whole-system analysis on the small scale, or
  approximate compositional analysis on the large scale. The simplest
  of cells is perhaps a transmission line, which may be modeled by
  input and output voltages $v_{in}$ and $v_{out}$, a maximum
  allowable drop in voltage $\delta_v$, a resistance $R$ which is a
  function of the medium and transmission distance, a current rating
  $I$, and a power rating $P$. Ignoring delays, one can describe such
  a cell by a set of constraints: {\em e.g.}, $v_{out}=v_{in} - R*I$
  (the voltage at the output is the difference between the input
  voltage and the voltage drop due to resistance), $v_{out}*I \leq P$
  (the power drain cannot exceed a maximum rated wattage), and $R*I
  \leq \delta_v$ (the drop in voltage must be less than what is
  allowed). Similarly, modules for other types of cells may be
  specified (or left unspecified as holes) and arrangements of such
  modules may be used to model large-scale smart grids, allowing
  designers to explore ``what if'' scenarios, {\em e.g.}, under what
  conditions would a hole in the grid cause a safety violation? or
  what are the most efficient settings ({\em e.g.}, power generation
  and routing decisions) in terms of power loss due to inefficient
  transmission? The introduction of ``smart'' computational processes
  in the grid ({\em e.g.}, feedback-based power management) and the
  expected diversity of technologies to be plugged into the grid make
  the consideration of such questions quite critical.

\noindent{\bf A Type Theory and Formal Semantics of Flow Networks.}
  Associated with our DSL is a {\em type theory}, a system of formal
  annotations to express desirable properties of flow networks
  together with rules that enforce them as {\em invariants} across
  their interfaces, {\em i.e.}, the rules guarantee the properties are
  preserved as we build larger networks from smaller ones.

A prerequisite for a type theory is a {\em formal semantics} -- a
  rigorous definition of the entities that qualify as feasible flows
  through the networks, possibly restricted to satisfy additional
  efficiency or safety requirements. This can be carried out in one of
  two ways, as a {\em denotational} semantics or as an {\em
  operational} (or {\em reduction}) semantics. In the first approach,
  a feasible flow through the network is denoted by a function, and
  the semantics of the network is the set of all such functions. In
  the second approach, the network is uniquely rewritten to another
  network in {\em normal form} (appropriately defined), and the
  semantics of the network is its normal form or directly extracted
  from it. Though the two can be shown to be equivalent (in a sense
  that can be made precise), whenever we need to invoke a network's
  semantics, we rely on the denotational definition in order to avoid
  complexity issues related to the operational definition. Some of
  these complexity issues are already evident from the form of network
  specifications we can write in our DSL.

As we alluded before, a distinctive feature of our DSL is the presence 
of \emph{holes} in network specifications, together with constructs of the form:
\(
  \Let{X}{= \M}{\N}
\),
  which informally says ``network $\M$ may be safely placed in the
  occurrences of hole $X$ in network $\N$''. What ``safely'' means
  will later depend on the invariant properties that typings are
  formulated to enforce. There are other useful \emph{hole-binders} 
  besides \textbf{\textsf{let-in}}, which we denote \textbf{\textsf{try-in}}, 
  \textbf{\textsf{mix-in}}, and \textbf{\textsf{letrec-in}}. An informal
  explanation of what these hole-binders mean is in 
  Remark~\ref{rem:other-derived-constructors} and
  Example~\ref{ex:illustrate-inductive-def}.

  Rewriting a specification in order to eliminate all occurrences of holes and
  hole-binders is a costly process, generally resulting in an exponential
  growth in the size of the expression denoting the specification,
  which poses particular challenges in the definition of an operational
  semantics.  We set up a typing system and prove its soundness for our DSL 
  without having to explicitly carry out such exponential-growth rewriting.

Our DSL provides two other primitive constructs, one of the form
  $\bigl(\ConnP{\M_1}{\M_2}\bigr)$ and another of the form
  $\Loop{\Angles{a,b}}{\N}$. The former juxtaposes two networks $\M_1$
  and $\M_2$ in parallel, and the latter binds the output arc $a$ of a
  network $\N$ to its input arc $b$. With these primitive or core
  constructors, we can define many others as \emph{derived}
  constructors and according to need.

\noindent{\bf Paper Overview and Context.}
  The remainder of this paper is organized as follows.
  Section~\ref{sect:basics} is devoted to preliminary definitions.
  Section~\ref{sect:inductive} introduces the syntax of our DSL and
  lays out several conditions for the well-formedness of network
  specifications written in it. We only include the
  \textbf{\textsf{let-in}} constructor, delaying the full treatment of
  \textbf{\textsf{try-in}}, \textbf{\textsf{mix-in}},
  \textbf{\textsf{letrec-in}}, to subsequent reports.

The formal semantics of flow networks are introduced in
  Section~\ref{sect:flows} and a corresponding type theory is
  presented in Section~\ref{sect:typings-are-polytopes}.  The type
  theory is syntax-directed, and therefore \emph{modular}, as it
  infers or assigns typings to objects in a stepwise inside-out
  manner. If the order in which typings are inferred for the
  constituent parts does not matter, we additionally say that the
  theory is \emph{fully compositional}. We add the qualifier ``fully''
  to distinguish our notion of compositionality from similar, but
  different, notions in other areas of computer
  science.\footnote{Adding to the imprecision of the word,
  ``compositional'' in the literature is sometimes used in the more
  restrictive sense of ``modular'' in our sense.}  We only include an
  examination of modular typing inference in this paper, leaving its
  (more elaborate) fully-compositional version to a follow-up report.

The balance of this paper expands on the fundamentals laid out in the
  first four sections: Sections~\ref{sect:inference}
  to~\ref{sect:relativized-typing-system} mostly deal with issues of
  typing inference, whether for the basic semantics of flow networks
  (introduced in Section~\ref{sect:flows}) or their relativized
  semantics, whereby flows are feasible if they additionally satisfy
  appropriately defined objective functions (introduced in
  Section~\ref{sect:relativized-semantics}).

\noindent{\bf Acknowledgment.}
  The work reported in this paper is a small fraction of a collective
  effort involving several people, under the umbrella of the
  \textbf{iBench Initiative} at Boston University. The reader is
  invited to visit the website \texttt{https://sites.google.com/site/ibenchbu/}
  for a list of
  participants, former participants, and other research activities.
The DSL presented in this paper, with its formal semantics and type
  system, is in fact a specialized and simpler version of a DSL we
  introduced earlier in our work for NetSketch, an integrated
  environment for the modeling, design and analysis of large-scale
  safety-critical systems with interchangeable parts
  \cite{BestavrosKfouryLapetsOcean:crts09,BestavrosKfouryLapetsOcean:hscc10,%
  SouleBestKfouryLapets:eoolt11}.
  In addition to its DSL, NetSketch has two other components currently
  under development: an automated verifier (AV), and a user interface
  (UI) that combines the DSL and the AV and adds appropriate tools for
  convenient interactive operation.

\section{Preliminary Definitions}
\label{sect:basics}
\label{sect:preliminary}

A \emph{small network} $\A$ is of the form $\A =
(\nn,\aaa)$ where $\nn$ is a set of nodes and $\aaa$ a set of directed
arcs. Capacities on arcs are determined by a lower-bound $L :
\aaa\to\nreals$ and an upper-bound $U : \aaa\to\nreals$ satisfying the
conditions $L(a) \leqslant U(a)$ for every $a\in\aaa$.  We write
$\reals$ and $\nreals$ for the sets of all reals and all non-negative
reals, respectively.
We identify the two ends of an arc $a\in\aaa$ by writing $\head{a}$
and $\tail{a}$, with the understanding that flow moves from $\tail{a}$
to $\head{a}$.  The set $\aaa$ of arcs is the disjoint union 
(denoted ``$\uplus$'') of three
sets: the set $\aaa_\text{\#}$ of internal arcs, the set
$\aaa_\text{in}$ of input arcs, and the set $\aaa_\text{out}$ of
output arcs:
\begin{alignat*}{3}
   &\aaa &&=\ \ 
   &&\aaa_\text{\#} \uplus \aaa_\text{in} \uplus \aaa_\text{out}\quad\text{where}
\\
   &\aaa_\text{\#} &&= 
   &&\Set{\,a\in\aaa\;|\;\head{a}\in\nn\text{ and }\tail{a}\in\nn\,}
\\
   &\aaa_\text{in} &&= 
   &&\Set{\,a\in\aaa\;|\;\head{a}\in\nn\text{ and }\tail{a}\not\in\nn\,}  
\\
   &\aaa_\text{out} &&= 
   &&\Set{\,a\in\aaa\;|\;\head{a}\not\in\nn\text{ and }\tail{a}\in\nn\,}
\end{alignat*}
The tail of an input arc, and the head of an 
output arc, are not attached to any node.
We do not assume $\A$ is connected as a directed graph -- a sensible
assumption in studies of network flows, whenever
there is only one input arc (or ``source node'') and one
output arc (or ``sink node'').  We assume 
$\nn\neq\varnothing$, \ie, there is at least one node in $\nn$,
without which there would be no input and no output arc, and nothing to
say.

A \emph{flow} $f$ in $\A$ is a function that assigns a non-negative
real to every $a\in\aaa$. Formally, a flow is a function $f :
\aaa\to\nreals$ which, if \emph{feasible}, satisfies ``flow
conservation'' and ``capacity constraints'' (below).

We call a bounded interval $[r,r']$ of reals, possibly
negative, a \emph{type}, and we call a \emph{typing} a function $T$
that assigns a type to every subset of input and output arcs.
Formally, $T$ is of the following form:%
    \footnote{Our notion of a ``typing'' as an assignment of
    types to the members of a powerset is different from
    a similarly-named notion in the study of type systems 
    for programming languages.
    In the latter, a typing refers to a derivable
    ``typing judgment'' consisting of  
    a program expression $M$, a type assigned to
    $M$, \emph{and} a type environment with a type
    for every free variable in $M$.} 
\[
    T\;:\ \power{\aaa_{\text{in}}\cup\aaa_{\text{out}}}\ \to\ \reals\times\reals
\]
where $\power{\ }$ is the power-set operator, \ie,
$\power{\aaa_{\text{in}}\cup\aaa_{\text{out}}} = \Set{A\,|\,A\subseteq
\aaa_{\text{in}}\cup\aaa_{\text{out}}}$.  As a function, $T$ is not
totally arbitrary and satisfies certain conditions, discussed in
Section~\ref{sect:notational}, which qualify it as a \emph{network
  typing}. Instead of writing $T(A) = \Angles{r,r'}$, where $A\subseteq
\aaa_{\text{in}}\cup\aaa_{\text{out}}$, we write $T(A) = [r,r']$.  We
do not disallow the possibility that $r > r'$ which will be an
empty type satisfied by no flow.

Informally, a typing $T$ imposes restrictions on a flow $f$ relative to
every $A\subseteq\aaa_{\text{in}}\cup\aaa_{\text{out}}$ which, if satisfied,
will guarantee that $f$ is feasible. Specifically, if $T(A) = [r,r']$, then
$T$ requires that the part of $f$ entering through the arcs in 
$A\cap\aaa_{\text{in}}$ minus the part of $f$ exiting through the arcs in 
$A\cap\aaa_{\text{out}}$ must be within the interval $[r,r']$.

\begin{remark}
Let $\A = (\nn,\aaa)$ be a small network.
We may want to identify some nodes as
\emph{producers} and some others as \emph{consumers}. In 
the presence of lower-bound and upper-bound functions $L$ and $U$, 
we do not need to do this explicitly. For example, if
$n$ is a node that produces an amount $r\in\nreals$, we  
introduce instead a new input arc $a$ entering $n$ with $L(a) = U(a) =
r$. Similarly, if $n'$ is a node that consumes an amount
$r'\in\nreals$, we introduce a new output arc $a'$ exiting
$n'$ with $L(a') = U(a') = r'$.  
The resulting network $\A'$ is equivalent to
$\A$, in that any feasible flow in $\A'$
induces a feasible flow in $\A$, and vice-versa.
\end{remark}

\paragraph{Flow Conservation, Capacity Constraints, Type Satisfaction.}
\label{sect:flow-conservation}

Though obvious, we precisely state fundamental concepts
underlying our entire examination and introduce some of our
notational conventions, in Definitions~\ref{def:flow-conservation},
\ref{def:capacity-constraints}, \ref{def:feasible-flows},
and~\ref{def:type-satisfaction}. 

\begin{definition}{Flow Conservation}
\label{def:flow-conservation}
If $A$ is a subset of arcs in $\A$ and $f$ a flow in $\A$, we write
$\sum f(A)$ to denote the sum of the flows assigned to all the arcs
in $A$:
\(    \sum f(A) = \sum \Set{f(a)\,|\,a\in A} 
\). By convention, $\sum\varnothing = 0$.
If $A= \Set{a_1,\ldots,a_p}$ is the set of all arcs entering
node $n$, and $B =\Set{b_1,\ldots,b_q}$ is the set of all arcs exiting
node $n$, then conservation of flow at $n$ is expressed by the 
linear equation:
\begin{equation}
\label{one}
   \sum\, f(A)\ = \ \sum\, f(B)  
\end{equation}
There is one such equation for every node {$n\in\nn$}.
\end{definition}

\begin{definition}{Capacity Constraints}
\label{def:capacity-constraints}
A flow $f$ satisfies the capacity constraints at arc $a\in \aaa$ if:
\begin{align}
\label{two}
      & L(a)\ \ \leqslant \ \ f(a) \ \leqslant \ \ U(a)
\end{align}
There are two such inequalities for every arc $a\in\aaa$.
\end{definition}

\begin{definition}{Feasible Flows}
\label{def:feasible-flows}
A flow $f$ is \emph{feasible} iff two conditions:
\begin{itemize} 
\item for every node $n\in\nn$, the equation in (\ref{one}) is satisfied,
\item for every arc $a\in\aaa$, the two inequalities in (\ref{two}) are satisfied, 
\end{itemize} 
following standard definitions of network flows.
\end{definition}

\begin{definition}{Type Satisfaction}
\label{def:type-satisfaction}
Let $T:\power{\aaa_{\text{in}}\cup\aaa_{\text{out}}}\to\reals\times\reals$ 
be a typing for the small network $\A$.
We say the flow $f$ \emph{satisfies} $T$ if, for every
$A\in\power{\aaa_{\text{in}}\cup\aaa_{\text{out}}}$ with $T(A) = [r,r']$,
it is the case:
\begin{align}
\label{three}
      &r\ \leqslant\quad
        \sum\, f(A\cap\aaa_{\text{in}})\ -\ \sum\, f(A\cap\aaa_{\text{out}})
        \quad \leqslant \ r'
\end{align}
We often denote a typing $T$ for $\A$ by simply writing $\A:T$.
\end{definition}

\section{DSL for Incremental and Modular Design of Flow Networks (Untyped)}
\label{sect:inductive}

The definition of small networks in Section~\ref{sect:basics}
was less general than our full definition of networks,
but it had the advantage of being more directly comparable
with standard graph-theoretic definitions. 
Our networks in general involve what we call ``holes''.  A \emph{hole}
$X$ is a pair $({\aaa}_{\text{in}},{\aaa}_{\text{out}})$ where
${\aaa}_{\text{in}}$ and ${\aaa}_{\text{out}}$ are disjoint finite sets of
input and output arcs. A hole $X$ is a place holder where networks can
be inserted, provided the \emph{matching-dimensions}
condition (in Section~\ref{sect:well-formedness}) is satisfied.

We use a BNF definition to generate formal expressions, each being
a formal description of a network. Such a formal expression
may involve subexpressions of the form:
\(
  \Let{X}{= \M}{\N}
\),
which informally says ``$\M$ may be safely placed in the occurrences of
hole $X$ in $\N$''. What ``safely'' means depends
on the invariant properties that typings are formulated to enforce.
In such an expression, we call the $X$ to the
left of ``$=$'' a \emph{binding} occurrence, and we call
all the $X$'s in $\N$ \emph{bound} occurrences.

If $\A = (\nn,\aaa)$ is a small network
where $\aaa = \aaa_{\text{\#}}\uplus\aaa_\text{in}\uplus\aaa_\text{out}$,
let $\inn{\A} = \aaa_{\text{in}}$, $\out{\A} = \aaa_{\text{out}}$,
and $\inter{\A} = \aaa_{\text{\#}}$.
Similarly, if $X = (\aaa_\text{in},\aaa_\text{out})$ is a hole,
let $\inn{X} = \aaa_{\text{in}}$, $\out{X} = \aaa_{\text{out}}$,
and $\inter{X} = \varnothing$.
We assume the arc names of small networks and holes
are all pairwise disjoint, \ie, every small network and every hole
has its own private set of arc names.

The formal expressions generated by our BNF
are built up from: the set of names for small
networks and the set of names for holes, using the constructors
$\ConnP{}{}$, $\textbf{\textsf{let-in}}$, and 
$\textbf{\textsf{bind}}$:
\begin{alignat*}{5}
&\A,\B,\C &&\in\textsf{\sc SmallNetworks} && && &&
\\
\nonumber
&X,Y,Z &&\in\textsf{\sc HoleNames} && && &&
\\
\nonumber
  &\M,\N,\PP &&\in\textsf{\sc Networks}\ &&::=
      \ &&\A   &&\text{small network name}
  \\
  &   &&  &&\ | &&X &&\text{hole name}
  \\
  &   &&  &&\ | &&\ConnP{\M}{\N} &&
         \text{parallel connection} 
  \\
  &   &&  &&\ | &&\Let{X}{=\M}{\N}
         \quad\ &&\text{let-binding of hole $X$}  
  \\
  &   &&  &&\ | &&\Loop{\Angles{a,b}}{\N} && 
         \text{bind $\head{a}$ to $\tail{b}$, where} 
  \\
  &  &&  && && &&\text{$\Angles{a,b}\in {\out{\N}}\times{\inn{\N}}$} 
  \end{alignat*}
where $\inn{\N}$ and $\out{\N}$ are the input and output arcs of
$\N$. In the full report~\cite{kfouryDSL:2011}, we formally define $\inn{\N}$ 
and $\out{\N}$, as well
as the set $\inter{\N}$ of internal arcs of $\N$, by structural induction.

\Hide
{
\begin{itemize}
\item If $\N$ is the name of small network $\A$,
      then \\ $\inn{\N} = \inn{\A}$,\ \ $\out{\N} = \out{\A}$,
      \ \ and\ \ $\inter{\N} = \inter{\A}$.
\item If $\N$ is the name of hole $X$,
      then \\ $\inn{\N} = \inn{X}$,\ \ $\out{\N} = \out{X}$,
      \ \ and\ \ $\inter{\N} = \varnothing$.
\item If $\N = (\ConnP{\M}{\M'})$,
      then \\ $\inn{\N} = \inn{\M}\cup\inn{\M'}$,
      \ \ $\out{\N} = \out{\M}\cup\out{\M'}$,
      \ \ and \ \ $\inter{\N} = \inter{\M}\cup\inter{\M'}$.
\item If $\N = (\Let{X}{=\M}{\N'})$,
      then \\ $\inn{\N} = \inn{\N'}$,\ \ $\out{\N} = \out{\N'}$,
      \ \ and\ \ $\inter{\N} = 
      \inter{\N'}\cup\inter{\M}$.
\item If $\N = \Loop{\Angles{a,b}}{\N'}$,
      then \\ $\inn{\N} = \inn{\N'}-\Set{b}$,
      \ \ $\out{\N} = \out{\N'}-\Set{a}$,\ \ and
      \ \ $\inter{\N} = \inter{\N'} \cup \Set{a}$ with $\head{a} := \tail{b}$.
\end{itemize}
}
We say a flow network $\N$ is \emph{closed} if every
hole $X$ in $\N$ is bound. We say $\N$ is \emph{totally closed} if
it is closed and $\inn{\N}=\out{\N}=\varnothing$, 
\ie, $\N$ has no input arcs and no output arcs.

\Hide{
\begin{remark}
A network specification $\N$, as defined by the BNF above, does not
introduce lower-bound and upper-bound capacities on arcs.  $\N$ only
defines a topology of a large network, starting from a collection of
small networks. From the capacities assigned to the small networks'
arcs, our typing theory will attempt to infer typings for all the
well-formed subparts (or subexpressions) of $\N$ and for $\N$
itself. If it succeeds to do this inference, the typings will certify
that the construction of every larger part from smaller parts respects
the invariant properties we wish to impart to all of $\N$.

Among invariant properties, we will want, at a minimum, that if there
are feasible flows in the smaller parts, then there are feasible flows
in the larger parts. There is one construction which already
illustrates this idea. In a subexpression of the form
$\Loop{\Angles{a,b}}{\N}$, if the lower-bound capacity of arc $a$ (or
$b$, resp.)  is strictly larger than the upper-bound capacity of arc
$b$ (or $a$, resp.), then there is no feasible flow in
$\Loop{\Angles{a,b}}{\N}$ and our typing theory will not allow this
``unsafe'' binding, \ie, connecting the head of $a$ to the
tail of $b$.
\end{remark}
}

\subsection{Derived Constructors}
\label{sect:derived-constructors}

From the three primitive constructors introduced above: $\ConnP{}{}$,
$\textbf{\textsf{let-in}}$, and 
$\textbf{\textsf{bind}}$, we can define several other
constructors. Below, we present four of these derived constructors precisely, and
mention several others in Remark~\ref{rem:other-derived-constructors}.
Our four derived constructors are used as in the following
expressions, where $\N$, $\N_i$, and $\M_j$, are network
specifications and $\theta$ is set of arc pairs:
\[ \Loop{\theta}{\N}
    \qquad \Conn{\theta}{\N_1}{\N_2}
    \qquad \ConnN{\N_1}{\N_2} 
    \qquad \Let{X}{\in\Set{{\M}_1,\ldots,{\M}_n}}{\N}
\]
The second above depends on the first, the third on the second, and the fourth
is independent of the three preceding it.
Let $\N$ be a network specification.  We write
$\theta\;\Partial{\out{\N}}{\inn{\N}}$ to denote a partial one-one map
from $\out{\N}$ to $\inn{\N}$. We may write the entries in $\theta$
explicitly, as in:
\[
    \theta\ =\ \Set{\Angles{a_1,b_1},\ldots,\Angles{a_k,b_k}}
\] 
where $a_1,\ldots,a_k\in \out{\N}$ and $b_1,\ldots,b_k\in \inn{\N}$.

Our first derived constructor is a generalization of
$\textbf{\textsf{bind}}$ 
and uses the same name. In this generalization of 
$\textbf{\textsf{bind}}$ 
the second argument is now $\theta$ as above rather than a single pair
$\Angles{a,b}\in {\out{\N}}\times{\inn{\N}}$. The expression
$\Loop{\theta}{\N}$ can be expanded as follows:
\[
      \Loop{\theta}{\N}\ \bm{\Longrightarrow}
        \ \Loop{\Angles{a_1,b_1}}
         {\Loop{\Angles{a_2,b_2}}
            {\ \cdots\ \Loop{\Angles{a_k,b_k}}{\N}\ \cdots\ }}
\]
where we first connect the head of $a_k$ to the tail of $b_k$ and
lastly connect the head of $a_1$ to the tail of $b_1$. A little
proof shows that the order in which we connect arc heads
to arc tails does not matter as far as our formal semantics and 
typing theory is concerned.
       
Our second derived constructor, called $\textbf{\textsf{conn}}$
(for ``connect''), uses the preceding generalization of 
$\textbf{\textsf{bind}}$ together with the constructor
$\ConnP{}{}$. Let $\N_1$ and $\N_2$ be network specifications,
and $\theta\;\Partial{\out{\N_1}}{\inn{\N_2}}$. We expand
the expression $\Conn{\theta}{\N_1}{\N_2}$ as follows:
\[
   \Conn{\theta}{\N_1}{\N_2}\ \ \bm{\Longrightarrow}
   \ \ \Loop{\theta}{(\ConnP{\N_1}{\N_2})}
\]
In words, $\textbf{\textsf{conn}}$ connects some of the output arcs
in $\N_1$ with as many input arcs in $\N_2$.

Our third derived constructor is a special case of the preceding
$\textbf{\textsf{conn}}$. Unless otherwise stated, we will assume
there is a fixed ordering of the input arcs and another fixed ordering
of the output arcs of a network. Let $\N_1$ be a network specification
where the number $m\geqslant 1$ of output arcs is exactly the number
of input arcs in another network specification $\N_2$, say:
\[
    \out{\N_1} = \Set{a_1,\ldots,a_m}
    \quad\text{and}\quad
    \inn{\N_2} = \Set{b_1,\ldots,b_m}
\]
where the entries in $\out{\N_1}$ and in $\inn{\N_2}$
are listed, from left to right, in their assumed ordering. Let
\[
   \theta = \ \Set{\Angles{a_1,b_1},\ldots,\Angles{a_m,b_m}}
          \ =\ \out{\N_1}\times\inn{\N_2}
\]
\ie, the first output arc $a_1$ of $\N_1$ is connected to the first
input arc $b_1$ of $\N_2$, the second output arc $a_2$ of $\N_1$ to
the second input arc $b_2$ of $\N_2$, etc. Our derived constructor
$(\ConnN{\N_1}{\N_2})$ can be expanded as follows:
\[
   (\ConnN{\N_1}{\N_2})\ \bm{\Longrightarrow}\ \Conn{\theta}{\N_1}{\N_2}
\]
which implies that $\inn{\ConnN{\N_1}{\N_2}} = \inn{\N_1}$ and 
$\out{\ConnN{\N_1}{\N_2}} = \out{\N_2}$. As expected, 
$\ConnN{}{}$ is associative as far as our formal semantics
and typing theory are concerned, 
\ie, the semantics and typings for $\ConnN{\N_1}{(\ConnN{\N_2}{\N_3})}$ 
and $\ConnN{(\ConnN{\N_1}{\N_2})}{\N_3}$ are the same.
 
A fourth derived constructor generalizes
$\textbf{\textsf{let-in}}$ and is expanded into several nested
$\textbf{\textsf{let}}$-bindings:
\[
   \bigl(\Let{X}{\in\Set{{\M}_1,\ldots,{\M}_n}}{\N}\bigr)\ \bm{\Longrightarrow} 
   \ \Bigl(\Let{X_1}{={\M}_1}{\bigl(\cdots\ \bigl(\Let{X_n}{={\M}_n}
    {({\N}_1\,\|\;\cdots\;\|\,{\N}_n)}\bigr)\ \cdots\bigr)}\Bigr)
\]
where $X_1,\ldots,X_n$ are fresh hole names and ${\N}_i$ is ${\N}$
with $X_i$ substituted for $X$, for every $1\leqslant i\leqslant n$.
Informally, this constructor says that \emph{every one} of the
networks $\Set{{\M}_1,\ldots,{\M}_n}$ can be ``safely''
placed in the occurrences of $X$ in ${\N}$.

\Hide{We will use these 4 derived constructors in
examples, but when setting up the typing rules for networks in
general, we will revert back to the primitive constructors.}

\Hide{\begin{remark}
\label{rem:order-of-construction}
If we consider graphical representations of constructions such as
$\Loop{\theta}{\N}$ and $\ConnN{\N_1}{\ConnN{\N_2}{\N_3}}$, then the
order in which we connect the arcs in the graphs does not matter,
obviously.  But we will invoke graphical representations only
informally.  To formally translate our network specifications into
graphical representations in some unique normal form -- which requires
not only expanding all derived constructors but also, more
challengingly, introducing formal rules to reduce all let-bindings --
is the basis of an \emph{operational} (or \emph{reduction}) approach
to the semantics of network specifications. However, this is something
we purposely avoid in this report, for reasons we futher elaborate in
Remark~\ref{rem:normal-form} below.

In preference to an operational approach, we choose a
\emph{denotational} approach to define the semantics of network
specifications, in Section~\ref{sect:flows}.
\end{remark}
}

\begin{remark}
\label{rem:other-derived-constructors}
Other derived constructors can be defined according to need in
applications. We sketch a few. An obvious generalization
of $\ConnN{}{}$ cascades the same network $\N$ some $n\geqslant 1$
times, for which we write $\ConnNN{\N}{n}$. A condition for
well-formedness is that $\N$'s input and output
dimensions must be equal.

Another derived constructor is $\Merge{\N_1}{\N_2}{\N_3}$ which
connects all the output arcs of $\N_1$ and $\N_2$ to all the input
arcs of $\N_3$. For well-formedness, this requires the output
dimensions of $\N_1$ and $\N_2$ to add up to the input dimension of
$\N_3$. And similarly for a derived constructor of the form
$\Fork{\N_1}{\N_2}{\N_3}$ which connects all the output arcs of $\N_1$
to all the input arcs of $\N_2$ and $\N_3$.

While all of the preceding derived constructors can be expanded using our
primitive constructors, not every constructor we may devise can be so
expanded. For example, a constructor of the form
\[
   \Try{X}{\in\Set{{\M}_1,\ldots,{\M}_n}}{\N\!}
\]
which we can take to mean that \emph{at least one} ${\M}_i$ can be
``safely'' placed in all the occurrences of $X$ in $\N$, cannot be
expanded using our primitives and the way we define their semantics in
Section~\ref{sect:flows}. Another constructor also requiring 
a more developed examination is of the form
\[
   \Mix{X}{\in\Set{{\M}_1,\ldots,{\M}_n}}{\N\!}
\]
which we can take to mean that every combination (or
mixture) of \emph{one or more} ${\M}_i$ can
be selected at the same time and ``safely'' placed
in the occurrences of $X$ in $\N$, generally placing different 
${\M}_i$ in different occurrences.
The constructors 
$\textbf{\textsf{try-in}}$ and $\textbf{\textsf{mix-in}}$ are examined
in a follow-up report. An informal understanding of how they
differ from the constructor $\textbf{\textsf{let-in}}$
can be gleaned from Example~\ref{ex:illustrate-inductive-def}.

Another useful constructor introduces recursively defined components
with (unbounded) repeated patterns. In its simplest form, it can
be written as:
\[
  \Letrec{X}{=\M[X]}{\N[X]}
\]
where we write $\M[X]$ to indicate that $X$ occurs free in $\M$, and similarly
in $\N$. Informally, this construction corresponds to placing an open-ended 
network of the form $\M[\M[\M[\cdots]]]$ in the occurrences of $X$ in $\N$. 
A well-formedness condition here is that the input and output 
dimensions of $\M$ must match those of $X$ .
We leave for future examination the semantics and typing of 
\textbf{\textsf{letrec-in}}, which are still more involved than those of
\textbf{\textsf{try-in}} and \textbf{\textsf{mix-in}}. 
\end{remark}

\subsection{Well-Formed Network Specifications}
\label{sect:well-formedness}

In the full report~\cite{kfouryDSL:2011}, we 
spell out 3 conditions, not enforced by the BNF definition at the
beginning of Section~\ref{sect:inductive}, which guarantee what we
call the \emph{well-formedness} of network specifications.  We call them:
\begin{itemize}
\item the \emph{matching-dimensions} condition, 
\item the \emph{unique arc-naming} condition, 
\item the \emph{one binding-occurrence} condition. 
\end{itemize}
These three conditions are automatically satisfied by small
networks. Although they could be easily incorporated into our
inductive definition, more than BNF style, they would 
obscure the relatively simple structure of our network
specifications. 

We only briefly explain what the second condition specifies: To avoid
ambiguities in the formal semantics of Section~\ref{sect:semantics}, 
we need to enforce in the specification of
a network $\N$ that no arc name refers to two different arcs. This in turn
requires that we distinguish the arcs of the different copies of the same hole
$X$. Thus, if we use $k\geqslant 2$ copies of $X$, we rename their arcs so that
each copy has its own set of arcs. We write $\rename{X}{1},\ldots,\rename{X}{k}$ 
to refer to these $k$ copies of $X$.
For further details on the \emph{unique arc-naming} condition, and full
explanation of the two other conditions, the reader is referred 
to~\cite{kfouryDSL:2011}.

\begin{example}
\label{ex:illustrate-inductive-def}
We illustrate several of the notions introduced so far. We use one
hole $X$, and 4 small networks: \textbf{\textsf{F}} (``fork''),
\textbf{\textsf{M}} (``merge''), $\A$, and $\B$. These will be used
again in later examples.  We do not assign lower-bound and upper-bound
capacities to the arcs of \textbf{\textsf{F}}, \textbf{\textsf{M}},
$\A$, and $\B$ -- the arcs of holes are never assigned capacities --
because they play no role before our typing theory is introduced.
Graphic representations of \textbf{\textsf{F}}, \textbf{\textsf{M}}, 
and $X$ are shown in Figure~\ref{fig:fork-merge-hole}, and of $\A$ and $\B$ in
Figure~\ref{fig:six-and-eight-naked}. A possible network specification
$\N$ with two bound occurrences of $X$ may read as follows:
\[
   \N\ =\ \ \Let{X}{\ \in\Set{\A,\B}\ }{
   \ \ \Conn{\theta_1}{\ \textbf{\textsf{F}}}
               {\ \Conn{\theta_2}{\ \rename{X}{1}}
                 {\ \Conn{\theta_3}{\ \rename{X}{2}}{\ \textbf{\textsf{M}}}}}}
\]
where $\theta_1 = \Set{\Angles{c_2,\rename{e_1}{1}},\Angles{c_3,\rename{e_2}{1}}}$,
$\theta_2 = \Set{\Angles{\rename{e_3}{1},\rename{e_1}{2}},
 \Angles{\rename{e_4}{1},\rename{e_2}{2}}}$, and 
$\theta_3 = \Set{\Angles{\rename{e_3}{2},d_1},
                       \Angles{\rename{e_4}{2},d_2}}$.
We wrote $\N$ above using some of the derived constructors
introduced in Section~\ref{sect:derived-constructors}. Note that:
\begin{itemize}
\item all the output arcs $\Set{c_2,c_3}$ of \textbf{\textsf{F}}
      are connected to all the input arcs 
      $\Set{\rename{e_1}{1},\rename{e_2}{1}}$ of $\rename{X}{1}$,
\item all the output arcs $\Set{\rename{e_3}{1},\rename{e_4}{1}}$ of 
      $\rename{X}{1}$ are connected to all the input arcs 
      $\Set{\rename{e_1}{2},\rename{e_2}{2}}$ of $\rename{X}{2}$,
\item all the output arcs $\Set{\rename{e_3}{2},\rename{e_4}{2}}$ of 
      $\rename{X}{2}$ are connected to all the input arcs 
      $\Set{d_1,d_2}$ of \textbf{\textsf{M}},
\end{itemize}
Hence, according to Section~\ref{sect:derived-constructors}, we can write
more simply:
\[
  \N\ =\ \ \Let{X}{\ \in\Set{\A,\B}\ }
    {\ \Bigl(   
       \ConnN{\textbf{\textsf{F}}}
         {\ConnN{\rename{X}{1}}
            {\ConnN{\rename{X}{2}}{\textbf{\textsf{M}}}}}   \Bigr)\ }
\]
with now $\inn{\N} = \Set{c_1}$ and $\out{\N} = \Set{d_3}$.
The specification $\N$ says that $\A$ or 
$\B$ can be selected for insertion
wherever hole $X$ occurs. 
Though we do not define the reduction of $\textbf{\textsf{let-in}}$-bindings 
formally, $\N$ can be viewed as representing two different
network configurations:
\[
   \N_1\ =\ \ \ConnN{\textbf{\textsf{F}}}
                     {\ConnN{\rename{\A}{1}}
                         {\ConnN{\rename{\A}{2}}{\textbf{\textsf{M}}}}} 
   \quad\text{and}\quad
   \N_2\ =\ \ \ConnN{\textbf{\textsf{F}}}
                     {\ConnN{\rename{\B}{1}}
                         {\ConnN{\rename{\B}{2}}{\textbf{\textsf{M}}}}} 
\]
We can say nothing here
about properties, such as safety, being satisfied or violated 
by these two configurations. The semantics of our
$\textbf{\textsf{let-in}}$ constructor later will be equivalent to
requiring that both configurations be ``safe'' to use.
By contrast, the  constructor $\textbf{\textsf{try-in}}$
mentioned in Remark~\ref{rem:other-derived-constructors}
requires only ${\N}_1$ or ${\N}_2$, but not necessarily both,
to be safe, and the constructor $\textbf{\textsf{mix-in}}$
additionally requires:
\[
   \N_3\ =\ \ \ConnN{\textbf{\textsf{F}}}
                     {\ConnN{\rename{\A}{1}}
                         {\ConnN{\rename{\B}{2}}{\textbf{\textsf{M}}}}} 
   \quad\text{and}\quad
   \N_4\ =\ \ \ConnN{\textbf{\textsf{F}}}
                     {\ConnN{\rename{\B}{1}}
                         {\ConnN{\rename{\A}{2}}{\textbf{\textsf{M}}}}} 
\]
to be safe. Safe substitution into holes according to 
$\textbf{\textsf{mix-in}}$ implies safe substitution according 
to $\textbf{\textsf{let-in}}$, which in turn implies safe substitution
according to $\textbf{\textsf{try-in}}$.
\end{example}


\begin{figure}[!ht] 
\begin{center}
\hspace*{.2in}
\begin{minipage}[b]{0.3\linewidth}
\includegraphics[scale=.22,trim=0cm 15.90cm 6.0cm 0cm,clip]{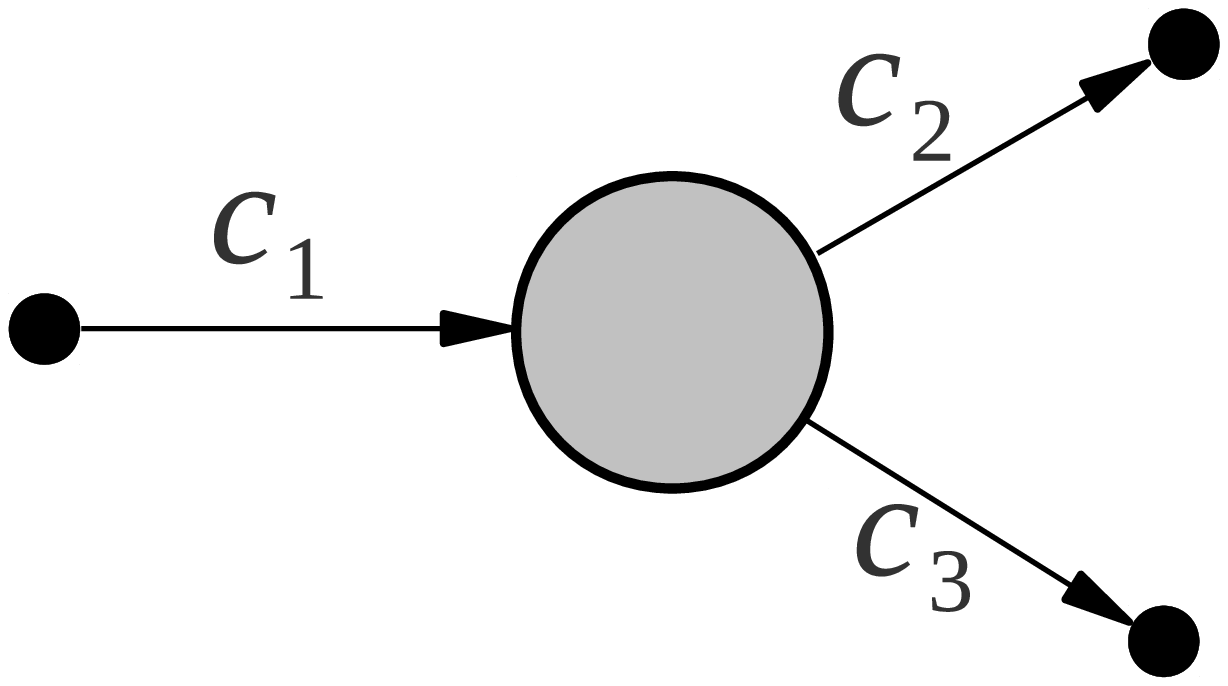}
\end{minipage}
\begin{minipage}[b]{0.3\linewidth}
\includegraphics[scale=.22,trim=0cm 15.90cm 2cm 0cm,clip]{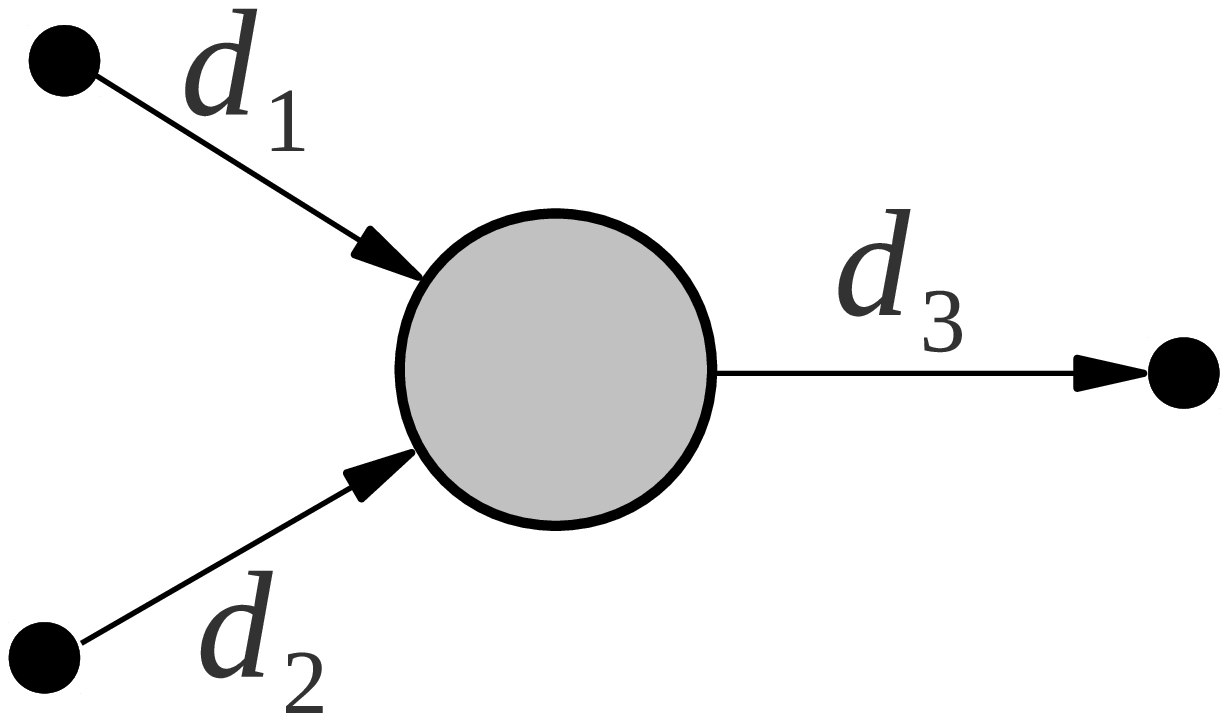}
\end{minipage}
\begin{minipage}[b]{0.3\linewidth}
\includegraphics[scale=.22,trim=0.7cm 13.50cm 0cm 3.5cm,clip]{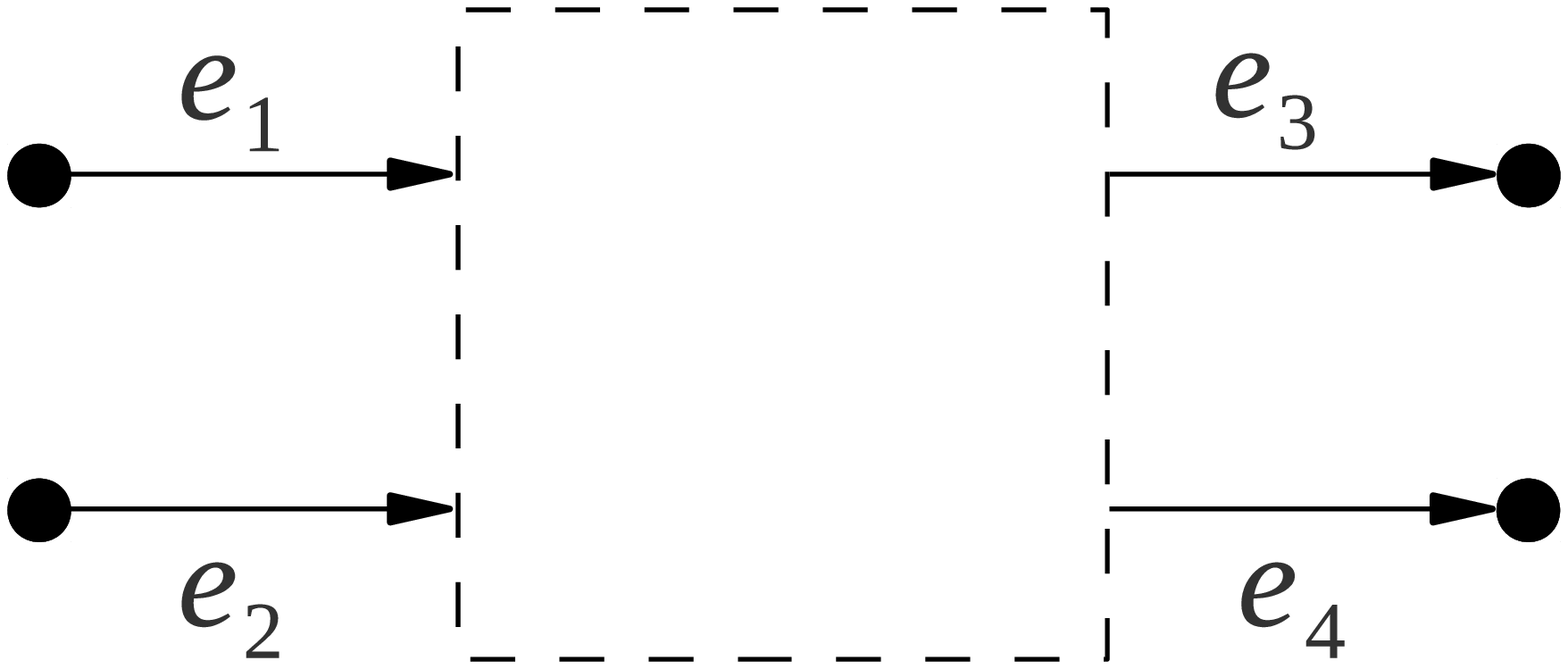}
\end{minipage}
\caption{Small network \textbf{\textsf{F}} (on the left), 
         small network \textbf{\textsf{M}} (in the middle), 
         and hole $X$ (on the right),
         in Example~\ref{ex:illustrate-inductive-def}.} %
\label{fig:fork-merge-hole}
\end{center}
\end{figure}

\vspace*{-.25in}
\begin{figure}[!ht] 
\begin{center}
\hspace*{.2in}
\begin{minipage}[b]{0.45\linewidth}
\includegraphics[scale=.3,trim=0cm 13.50cm 0cm 0cm,clip]{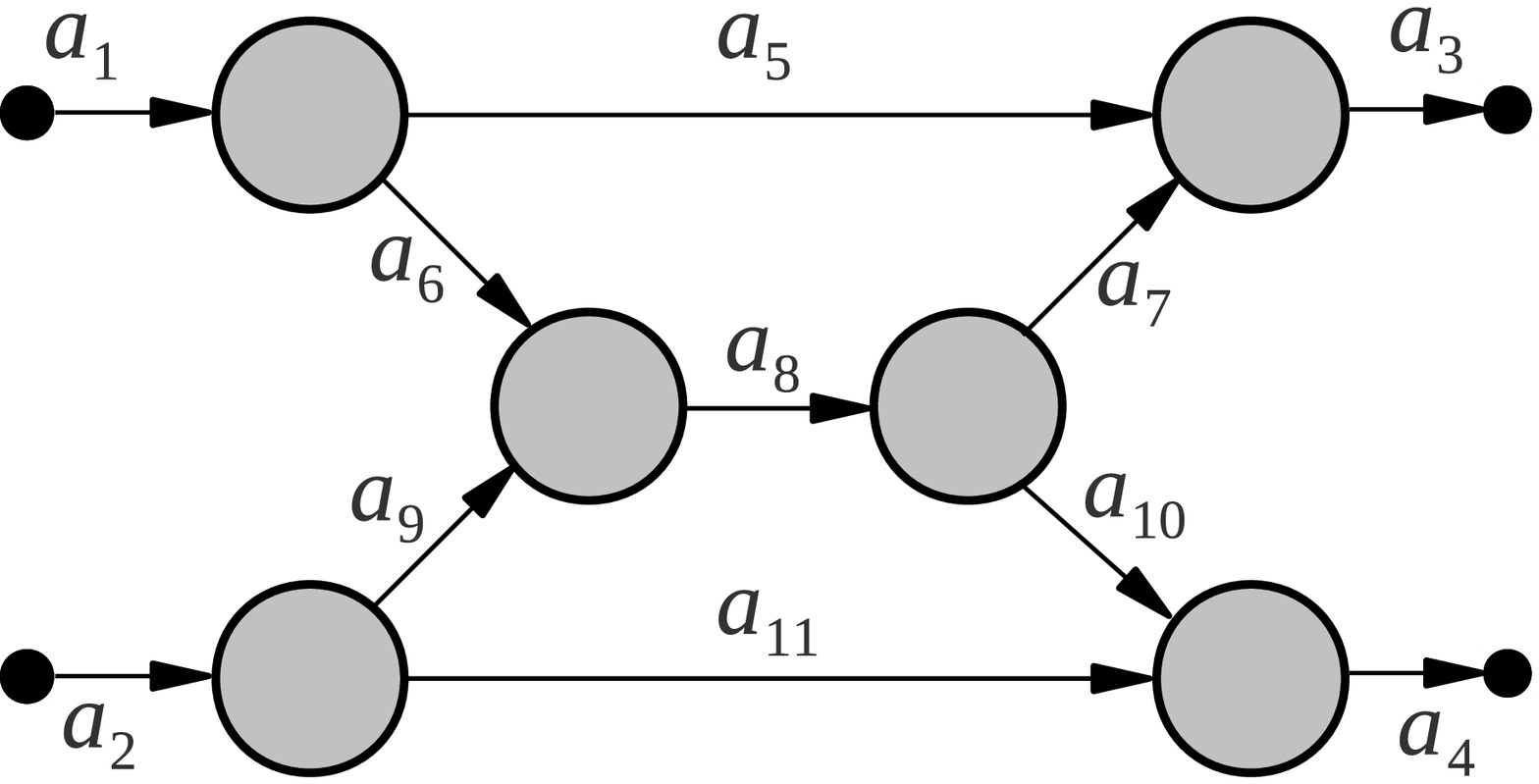}
\end{minipage}
\begin{minipage}[b]{0.45\linewidth}
\includegraphics[scale=.3,trim=0cm 11.0cm 0cm 0cm,clip]{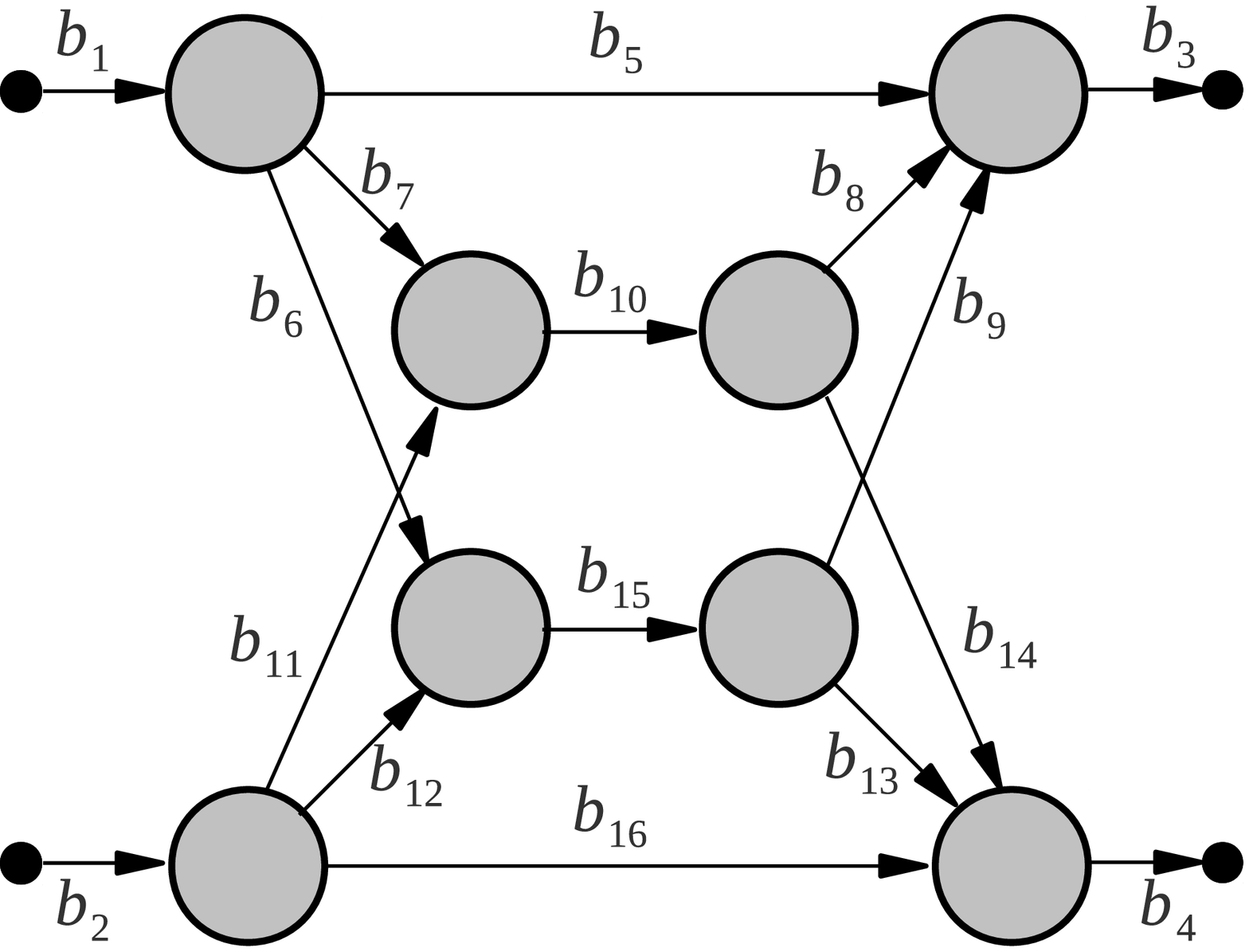}
\end{minipage}
\caption{Small networks $\A$ (on the left) and $\B$ (on the right)
         in Example~\ref{ex:illustrate-inductive-def}.} %
\label{fig:six-and-eight-naked}
\end{center}
\end{figure}

\vspace*{-.12in}


\section{Formal Semantics of Flow Networks}
\label{sect:flows}
\label{sect:semantics}

The preceding section explained what we need to write to specify a
network formally. Let $\N$ be such a network specification. By
well-formedness, every small network $\A$ appearing in $\N$ has its
own separate set of arc names, and every bound occurence
$\rename{X}{i}$ of a hole $X$ also has its own separate set of arc
names, where $i\geqslant 1$ is a renaming index. (Renaming indeces 
are defined in Section~\ref{sect:well-formedness}.)
With every small network $\A$, we associate two sets of functions,
its \emph{full semantics} $\fullSem{\A}$ and its \emph{IO-semantics} 
$\ioSem{\A}$. Let $\aaa_{\text{in}}=\inn{\A}$, $\aaa_{\text{out}}=\out{\A}$,
and $\aaa_{\text{\#}}=\inter{\A}$. 
The sets $\fullSem{\A}$ and $\ioSem{\A}$ are defined thus:
\begin{alignat*}{3}
  &\fullSem{\A}\ &&= 
   \ &&\Set{\,f:\aaa_{\text{in}}\uplus\aaa_{\text{out}}\uplus\aaa_{\text{\#}}
        \to\nreals\;|\;\text{$f$ is a feasible flow in $\A$}\,}
\\
  &\ioSem{\A}\ &&=
   \ &&\Set{\,f:\aaa_{\text{in}}\uplus\aaa_{\text{out}}\to\nreals\;|
   \;\text{$f$ can be extended to a feasible flow $f'$ in $\A$}\,}
\end{alignat*}
%
Let $X$ be a hole, with $\inn{X} = \aaa_{\text{in}}$ and
$\out{X} = \aaa_{\text{out}}$. The \emph{full semantics} $\fullSem{X}$ and 
the \emph{IO-semantics} $\ioSem{X}$ are the same set of functions:
\[  
   \fullSem{X}\ =\ \ioSem{X}\ \subseteq
   \ \Set{ f:\aaa_{\text{in}}\uplus\aaa_{\text{out}}\to\nreals\;|
   \; \text{$f$ is a bounded function}\,}
\]
This definition of $\fullSem{X} = \ioSem{X}$ is ambiguous: In contrast 
to the uniquely defined full semantics and IO-semantics of a small network $\A$, 
there are infinitely many $\fullSem{X} = \ioSem{X}$ for the same $X$,
but exactly one (possibly $\fullSem{X} = \ioSem{X} = \varnothing$)
will satisfy the requirement in clause 4 below. 
 
Starting from the full semantics of small networks and holes, we
define by induction the full semantics $\fullSem{\N}$ of a network
specification $\N$ in general. In a similar way, we
can define the IO-semantics $\ioSem{\N}$ of $\N$ by induction,
starting from the IO-semantics of small networks and holes.  For
conciseness, we define $\fullSem{\N}$ separately first, and
then define $\ioSem{\N}$ from $\fullSem{\N}$. We need a few
preliminary notions.
Let $\M$ be a network specification. By
our convention of listing all input arcs first, all output arcs second,
and all internal arcs third, let: 
\[
   \inn{\M} = \Set{a_1,\ldots,a_k},
   \quad
   \out{\M} = \Set{a_{k+1},\ldots,a_{k+\ell}}, 
   \quad\text{and}
   \ \ \inter{\M} = \Set{a_{k+\ell+1},\ldots,a_{k+\ell+m}}.
\]
If $f\in\fullSem{\M}$ with $f(a_1) = r_1,\ldots,f(a_{k+\ell+m}) = r_{k+\ell+m}$,  
we may represent $f$ by the sequence
$\Angles{r_1,\ldots,r_{k+\ell+m}}$. We may therefore represent:
\begin{itemize}
\item $\rest{f}{\inn{\M}}$ by the sequence $\Angles{r_1,\ldots,r_k}$,  
\item $\rest{f}{\out{\M}}$ by the sequence 
      $\Angles{r_{k+1},\ldots,r_{k+\ell}}$, and
\item $\rest{f}{\inter{\M}}$ by the sequence 
      $\Angles{r_{k+\ell+1},\ldots,r_{k+\ell+m}}$,
\end{itemize}
where $\rest{f}{\inn{\M}}$, $\rest{f}{\out{\M}}$, and
$\rest{f}{\inter{\M}}$, are the restrictions of $f$
to the subsets $\inn{\M}$, $\out{\M}$, and $\inter{\M}$, of its
domain. Let $\N$ be another network specification and
$g\in\fullSem{\N}$. We define $\ConnPT{f}{g}$ as follows:
\[
   (\ConnPT{f}{g}) = 
   \rest{f}{\inn{\M}}\cdot\rest{g}{\inn{\N}}\cdot
   \rest{f}{\out{\M}}\cdot\rest{g}{\out{\N}}\cdot
   \rest{f}{\inter{\M}}\cdot\rest{g}{\inter{\N}} 
\]
where ``$\cdot$'' is sequence concatenation. 
The operation ``$\ConnPT{\!}{\!}$'' on flows is associative, but not
commutative, just as the related constructor ``$\ConnP{\!}{\!}$'' on
network specifications.
%
We define the full semantics $\fullSem{\M}$ for every subexpression
$\M$ of $\N$, by induction on the structure of the specification $\N$:
\begin{enumerate}
\item If $\M = \A$, then $\fullSem{\M} = \fullSem{\A}$.
\item If $\M = \rename{X}{i}$, then 
      $\fullSem{\M} = \rename{\fullSem{X}}{i\,}$.
\item If $\M = \bigl(\ConnP{\PP_1}{\PP_2}\bigr)$,
      then $\fullSem{\M}\ = 
      \ \SET{\,(\ConnPT{f_1}{f_2})\;\bigl|\;f_1\in\fullSem{\PP_1}\text{ and }
      f_2\in\fullSem{\PP_2}\,}$.
\item If $\M = \bigl(\Let{X}{=\PP}{{\PP}'}\bigr)$,
      then $\fullSem{\M} = \fullSem{{\PP}'}$, 
      provided two conditions:%
        \footnote{\label{crucial-foot}``$\dimIO{X} \approx \dimIO{\PP}$''
          means the number of input arcs and their ordering (or input dimension)
          \emph{and} the number of output arcs and their ordering (or output dimension)
          of $X$ match those of $\PP$, up to arc renaming (or dimension renaming).
          Similarly, ``$\fullSem{X} \approx \Set{\rest{g}{A}|g\in\fullSem{\PP}}$''
          means for every $f:\inn{X}\uplus\out{X}\to\nreals$, it holds that
          $f\in\fullSem{X}$ iff there is $g\in\fullSem{\PP}$
          such that $f \approx \rest{g}{A}$, where $\rest{g}{A}$ is the
          restriction of $g$ to the subset $A$ of its domain.}
      \begin{enumerate}
      \item
      $\dimIO{X}\ \approx\ \dimIO{\PP}$, 
      \item 
      $\fullSem{X}\ \approx\ \Set{\,\rest{g}{A}\;|\;g\in\fullSem{\PP}\,}$
       where $A = \inn{\PP}\cup\out{\PP}$.
      \end{enumerate}
\item If $\M = \Loop{\Angles{a,b}}{\PP}$, then 
      $\fullSem{\M}\ =\ \SET{\,f\;\bigl|\;f\in\fullSem{\PP}
                        \text{ and }f(a) = f(b)\,}$.
\end{enumerate}
All of $\N$ is a special case of a subexpression of $\N$, so that a
the semantics of $\N$ is simply $\fullSem{\N}$.  Note, in clause 2, that
all bound occurrences $\rename{X}{i}$ of the same hole $X$ are
assigned the same semantics $\fullSem{X}$, up to renaming of arc
names. We can now define the IO-semantics of $\N$ as follows:
\[
   \ioSem{\N}\ =\ \SET{\,\rest{f}{A}\;\bigl|\; f\in \fullSem{\N}\,}
\]
where $A = \inn{\N}\cup\out{\N}$ and $\rest{f}{A}$ is the 
restriction of $f$ to $A$. 

{
\begin{remark}
\label{rem:capacities}
For every small network $\A$ appearing in a network specification
$\N$, the lower-bound and upper-bound functions, $L_{\A}$ and
$U_{\A}$, are already defined. The lower-bound and upper-bound for all
of $\N$, denoted $L_{\N}$ and $U_{\N}$, are then assembled from those
for all the small networks. However, we do not need to explicitly
define $L_{\N}$ and $U_{\N}$ at every step of the inductive definition
of $\N$.

In clause 4, the lower-bound and upper-bound capacities on
an input/output arc $a$ of the hole $X$ are determined by those on
the corresponding arc, say $a'$, in ${\PP}$. Specifically,
$L_{X}(a)=L_{\PP}(a')$ and $U_{X}(a)=U_{\PP}(a')$. 
In clause 5, the lower-bound and upper-bound are implicitly set. 
Specifically, consider output arc $a$ and input arc $b$ in $\PP$,
with $L_{\PP}$ and $U_{\PP}$ already defined on $a$ and $b$. If
$\M = \Loop{\Angles{a,b}}{\PP}$, then:
\begin{alignat*}{3}
   &L_{\M}(a)\ &&=\ &&\max\ \Set{L_{\PP}(a),L_{\PP}(b)} 
\\
   &U_{\M}(a)\ &&=\ &&\min\ \Set{U_{\PP}(a),U_{\PP}(b)}  
\end{alignat*}
which are implied by the requirement that $f(a) = f(b)$. In $\M$, arc
$a$ is now internal and arc $b$ is altogether omitted.  On all the arcs
other than $a$, $L_{\M}$ and $U_{\M}$ are identical to $L_{\PP}$ and
$U_{\PP}$, respectively.  
\end{remark}
}

\Hide
{
\begin{remark}
\label{rem:special-cases}
We do not disallow the possibility that $\fullSem{\N} = \varnothing$,
which happens when there are no feasible flows in $\N$. For example,
if $\N$ mentions only one small network $\A$ and there are no 
feasible flows in $\A$, then it must be that
$\fullSem{\A} = \fullSem{\N} = \varnothing$, which also implies
$\ioSem{\A} = \varnothing$.

Another possibility is $\ioSem{\N} = \Set{\varnothing}$, different
from the preceding, the result of inductively defining a
\emph{totally closed} $\N$, \ie, $\N$ has no input arcs and no output
arcs. In such a case, there are \emph{internal} feasible flows only.

There are other special cases, when $\inn{\N} = \varnothing$ or
$\out{\N} = \varnothing$, but not both. For example, if $\inn{\N} =
\Set{a_1,\ldots,a_k}$ for some
$k\geqslant 1$ and $\out{\N} = \varnothing$, 
then either $\ioSem{\N} = \varnothing$ (there are no
feasible flows in $\N$) or $\ioSem{\N} = \Set{ f }$ with 
$f(a_1) =\cdots= f(a_k) = 0$.
\end{remark}
}

\begin{remark}
\label{rem:invariance-of-semantics}
We can define rewrite rules on network
specifications in order to reduce each into an equivalent finite set
of network specifications in \emph{normal form}, a normal form
being free of \textbf{\textsf{try-in}} bindings. We can do this so that the
formal semantics of network specifications are an \emph{invariant} of
this rewriting. This establishes the \emph{soundness} of
the \emph{operational semantics} (represented by
the rewrite rules) of our DSL relative to the formal semantics defined
above. We avoid formulating and presenting such rewriting rules in
this report, for reasons alluded to in the Introduction and again
in the last section.
\end{remark}

\paragraph{Flow Conservation, Capacity Constraints, 
Type Satisfaction (Continued).}
\label{sect:flow-conservation-continued}
The fundamental concepts stated in relation to small networks
$\A$ in Definitions~\ref{def:flow-conservation},
\ref{def:capacity-constraints}, and~\ref{def:feasible-flows},
are extended to arbitrary network specifications $\N$. These are
stated as ``properties'' (not ``definitions'') because they apply
to $\fullSem{\N}$ (not to $\N$), and $\fullSem{\N}$ is built up 
inductively from $\Set{\fullSem{\A}\,|\,\text{$\A$ occurs in $\N$}}$.

\begin{property}[Flow Conservation -- Continued]
\label{def:flow-conservation-bis}
\label{prop:flow-conservation-bis}
The nodes of $\N$ are all the nodes in
the small networks occurring in $\N$, because our DSL  
in Section~\ref{sect:inductive} does not introduce new nodes
beyond those in the small networks. Hence, $\fullSem{\N}$ satisfies 
\emph{flow conservation} because, for every small network $\A$ in 
$\N$, every $f\in\fullSem{\A}$ satisfies flow conservation at 
every node, \ie, the equation in~(\ref{one}) in
Definition~\ref{def:flow-conservation}.
\end{property}

\begin{property}[Capacity Constraints -- Continued] 
\label{def:capacity-constraints-bis}
\label{prop:capacity-constraints-bis}
The arcs introduced by our DSL,
beyond the arcs in the small networks,
are the input/output arcs of the holes. Lower-bound and
upper-bound capacities on the latter arcs are set in order
not to conflict with those already defined on the input/output
arcs of small networks. Hence, $\fullSem{\N}$ satisfies
\emph{the capacity constraints} because, for every small network $\A$
in $\N$, every $f\in\fullSem{\A}$ satisfies the capacity
constraints on every arc, \ie, the inequalities in~(\ref{two})
in Definition~\ref{def:capacity-constraints}.
\end{property}

However, stressing the obvious, even if $\fullSem{\A}\neq\varnothing$
for every small network $\A$ in $\N$, it may still be that $\N$ is unsafe
to use, \ie, it may still be that there is no feasible flow in $\N$ because 
$\fullSem{\N} = \varnothing$. We use the type system 
(Section~\ref{sect:typing-rules}) to reject unsafe network 
specifications $\N$. 

\begin{definition}{Type Satisfaction -- Continued}
\label{def:type-satisfaction-bis}
Let $\N$ be a network, with
$\aaa_{\text{in}} = \inn{\N}$, $\aaa_{\text{out}} = \out{\N}$,
and $\aaa_{\text{\#}} = \inter{\N}$.
A typing $T$ for $\N$, also denoted $(\N:T)$, is a function
\[
   T:\power{\aaa_{\text{in}}\cup\aaa_{\text{out}}}\to\reals\times\reals
\]
which may, or may not, be satisfied by $f\in\ioSem{\N}$
or by $f\in\fullSem{\N}$. We say
$f\in\ioSem{\N}$ or $f\in\fullSem{\N}$
\emph{satisfies} $T$ iff, for every
$A\subseteq {\aaa_{\text{in}}\cup\aaa_{\text{out}}}$ with 
$T(A) = [r,r']$, it is the case that:
\begin{equation}
\label{three-bis}
        r\ \leqslant\quad
        \sum\, f(A\cap\aaa_{\text{in}})\ -\ \sum\, f(A\cap\aaa_{\text{out}})
        \quad \leqslant \ r'
\end{equation}
The inequalities in~(\ref{three-bis}) extend those 
in~(\ref{three}) in Definition~\ref{def:type-satisfaction}
to network specifications in general. 
\Hide{
One special case not covered by~(\ref{three-bis}) is when
$\fullSem{\N} =\varnothing$, \ie, there are no feasible flows in $N$,
in which case also $\ioSem{\N}=\varnothing$.  An appropriate typing
$T$ for $\N$ in this case assigns the empty interval to every
$A\subseteq {\aaa_{\text{in}}\cup\aaa_{\text{out}}}$.
}
\Hide
{
We may restrict attention to the input arcs or to the output arcs, 
in which case we say $f\in\ioSem{\N}$ or $f\in\fullSem{\N}$ \emph{satisfies}
the typing $T$ \emph{at the input} or \emph{at the output}, respectively.
If the restriction is to the input, then for  
every $A\subseteq {\aaa_{\text{in}}}$ with $\inT{T}(A) = [r,r']$, we have:
\begin{align}
\label{four}
      & r \ \leqslant\quad
        \sum\, f(A)\quad 
        \leqslant \ r'
\end{align}
or if it is to the output, then for every $B\subseteq {\aaa_{\text{out}}}$ 
with $\outT{T}(B) = [s,s']$, we have:
\begin{align}
\label{five}
      & s \ \leqslant\quad
        -\sum\, f(B)\quad 
        \leqslant \ s'
\end{align}
If $f\in\ioSem{\N}$ or  $f\in\fullSem{\N}$
satisfies $T$ at the input, we may say $f$
\emph{satisfies} $\inT{T}$, and if at the output, we say $f$
\emph{satisfies} $\outT{T}$.
}
\end{definition}

\Hide
{
It is worth stressing that, while satisfaction of~(\ref{three-bis})
implies satisfaction of both~(\ref{four}) and~(\ref{five}), the
converse is not necessarily true: It may happen that $f$
satisfies~(\ref{four}) and~(\ref{five}) but not~(\ref{three-bis}). 
Example~\ref{ex:six-and-eight-node-networks} below illustrates this
point.
}

\Hide
{
\begin{remark}
\label{rem:commodities-bis}
If there is a set $K$ of several commodities, then~(\ref{three-bis}) 
must involve summations over $K$, as follows:
\begin{equation*}
        r\ \leqslant\quad
        \sum_{\kappa\in K}\sum\, f_{\kappa}(A\cap\aaa_{\text{in}})\ -
        \ \sum_{\kappa\in K}\sum\, f_{\kappa}(A\cap\aaa_{\text{out}})
        \quad \leqslant \ r'
\end{equation*}
where $f_{\kappa}$ is the flow of commodity $\kappa$. As mentioned already
(Remark~\ref{rem:commodities}), nothing essential is lost by
restricting attention to one commodity.
\end{remark}
}

\section{Typings Are Polytopes}
\label{sect:notational}
\label{sect:typings-are-polytopes}

Let $\N$ be a network specification, 
and let $\aaa_{\text{in}} = \inn{\N}$ and $\aaa_{\text{out}} = \out{\N}$.
Let $T$ be a typing for $\N$ that assigns an interval $[r,r']$ to 
$A\subseteq\aaa_{\text{in}}\cup\aaa_{\text{out}}$.  Let
$\size{\aaa_{\text{in}}}+\size{\aaa_{\text{out}}} = m$, for some
$m\geqslant 0$.  As usual, there is a fixed ordering on the 
arcs in $\aaa_{\text{in}}$ and again on the arcs in $\aaa_{\text{out}}$. 
With no loss of generality, suppose:
\[
   A_1 = A\cap\aaa_{\text{in}} = \Set{a_1,\ldots,a_k}
   \quad\text{and}\quad
   A_2 = A\cap\aaa_{\text{out}} = \Set{a_{k+1},\ldots,a_\ell},
\]
where $\ell \leqslant m$.  
Instead of writing $T(A) = [r,r']$, we may write:
\[
     T(A):\quad a_1+\cdots+a_k-a_{k+1}-\cdots -a_\ell\ :\ [r,r']
\]
where the inserted polarities, $+$ or $-$, indicate whether the
arcs are input or output, respectively. 
A flow through the arcs
$\Set{a_1,\ldots,a_k}$ contributes a \emph{positive} quantity, and
through the arcs $\Set{a_{k+1},\ldots,a_\ell}$ a \emph{negative}
quantity, and these two quantities together should add up to a
value within the interval $[r,r']$.

A typing $T$ for
$\aaa_{\text{in}}\cup\aaa_{\text{out}}$ induces a \emph{polytope} (or
\emph{bounded polyhedron}), which we call $\poly{T}$, in the Euclidean
hyperspace ${\reals}^m$.
We think of the $m$ arcs in $\aaa_{\text{in}}\cup\aaa_{\text{out}}$ as
the $m$ dimensions of the space ${\reals}^m$.  
$\poly{T}$ is the non-empty intersection of at most $2\cdot(2^m -1)$ halfspaces,
because there are $(2^m -1)$ non-empty subsets in
$\power{\aaa_{\text{in}}\cup\aaa_{\text{out}}}$. The 
interval $[r,r']$, which $T$ assigns to such a subset
$A = \Set{a_1,\ldots,a_{\ell}}$ as above, induces two linear
inequalities in the variables $\Set{a_1,\ldots,a_\ell}$, denoted
$T_{\geqslant}(A)$ and $T_{\leqslant}(A)$:
\begin{equation}
\label{sixA}
\text{$T_{\geqslant}(A)$:}\quad 
   a_1+\cdots+a_k-a_{k+1}-\cdots -a_\ell \,\geqslant\,r
\qquad\text{and}\qquad
\text{$T_{\leqslant}(A)$:}\quad 
   a_1+\cdots+a_k-a_{k+1}-\cdots -a_\ell\,\leqslant\,r'
\end{equation}
and, therefore, two halfspaces $\half{T_{\geqslant}(A)}$ and 
$\half{T_{\leqslant}(A)}$:
\begin{equation}
\label{sixB}
  \half{T_{\geqslant}(A)}\ =\ \Set{\,\bm{r}\in {\reals}^m\;|
  \; \bm{r} \text{ satisfies $T_{\geqslant}(A)$}\,} 
\qquad\text{and}\qquad
  \half{T_{\leqslant}(A)}\ =\ \Set{\,\bm{r}\in {\reals}^m\;|
  \; \bm{r} \text{ satisfies $T_{\leqslant}(A)$}\,} 
\end{equation}
We can therefore define $\poly{T}$ formally as follows:

\bigskip
\fbox{
\hspace*{-.3in}
\begin{minipage}{.72\textwidth}
\vspace*{-.2in}
\[
   \poly{T}\ =\  \bigcap\,\SET{\,\half{T_{\geqslant}(A)}\;\cap
     \;\half{T_{\leqslant}(A)}\;\bigl|
     \;\varnothing\neq A\subseteq \aaa_{\text{in}}\cup\aaa_{\text{out}}\,}
\]
\end{minipage}
}

\bigskip
\noindent
Generally, many of the inequalities induced by the typing $T$ will
be redundant, and the induced $\poly{T}$ will be defined by 
far fewer than $2\cdot (2^m -1)$ halfspaces. 

\subsection{Uniqueness and Redundancy in Typings}
\label{sect:uniqueness-and-redundancy}

We can view a network typing $T$ as a syntactic expression, with its
semantics $\poly{T}$ being a polytope in Euclidean hyperspace. As in
other situations connecting syntax and semantics, there are generally
distinct typings $T$ and $T'$ such that $\poly{T} = \poly{T'}$. This
is an obvious consequence of the fact that the same polytope can be
defined by many different equivalent sets of linear inequalities,
which is the source of some complications when we combine two typings 
to produce a new one.

To achieve uniqueness of typings, as well as some efficiency of
manipulating them, we may try an approach that eliminates redundant
inequalities in the collection:
\begin{equation}
\label{sixC}
   \Set{\,T_{\geqslant}(A)\;|
   \;\varnothing\neq A\in\power{\aaa_{\text{in}}\cup\aaa_{\text{out}}}\,}
   \ \cup
   \ \Set{\,T_{\leqslant}(A)\;|
   \;\varnothing\neq A\in\power{\aaa_{\text{in}}\cup\aaa_{\text{out}}}\,}
\end{equation}
where $T_{\geqslant}(A)$ and $T_{\leqslant}(A)$ are as in~(\ref{sixA})
above.  There are standard procedures which determine whether a finite
set of inequalities are linearly independent and, if they are not,
select an equivalent subset of linearly independent inequalities.
\Hide
{
However, even if we agree on a canonical order in
which to carry out the elimination of redundant inequalities, the same
inequality retained at the end can be written in different forms, \eg,
$2a_1 - (1/2) a_2 \leqslant 4$ is equivalent to $a_1 - (1/4)a_2 \leqslant 2$
and $4a_1 - a_2 \leqslant 8$ and many others. We therefore need
to apply some care when using such elimination procedures, if we 
want to uniquely produce non-redundant typings. 
}
Some of these issues are 
taken up in the full report~\cite{kfouryDSL:2011}.

If $\N_1 : T_1$ and $\N_2 : T_2$ are typings for networks $\N_1$ and
$\N_2$ with matching input and output dimensions, we write $T_1\equiv
T_2$ whenever $\poly{T_1} \approx \poly{T_2}$, in which case we say that 
$T_1$ and $T_2$ are \emph{equivalent}.%
   \footnote{``$\poly{T_1} \approx \poly{T_2}$'' means that 
   $\poly{T_1}$ and $\poly{T_2}$ are the same up to renaming their
   dimensions, \ie, up to renaming the input and output arcs in $\N_1$
   and $\N_2$.}
If $\N_1 = \N_2$, then $T_1\equiv
T_2$ whenever $\poly{T_1} = \poly{T_2}$.

\begin{definition}{Tight Typings}
\label{def:tight-typings}
Let $\N$ be a network specification, with $\aaa_{\text{in}} =
\inn{\N}$ and $\aaa_{\text{out}} = \out{\N}$, and
$T :\power{\aaa_{\text{in}}\cup\aaa_{\text{out}}}
\to \reals\times\reals$ a typing for $\N$. 
$T$ is a \emph{tight} typing if for
every typing $T'$ such that $T\equiv T'$ and
for every $A\subseteq \aaa_{\text{in}}\cup\aaa_{\text{out}}$,
the interval $T(A)$ is contained in the interval $T'(A)$, \ie,
$T(A)\subseteq T'(A)$.
\end{definition}

\begin{proposition}[Every Typing Is Equivalent to a Tight Typing]
\label{prop:converting-to-uniformly-tight}
There is an algorithm $\tight{}$ which, given
a typing $(\N:T)$ as input, always terminates and
returns an equivalent tight typing $(\N:\tight{T})$.%
\end{proposition}

\subsection{Valid Typings and Principal Typings}
\label{sect:valid-vs-principal}

Let $\N$ be a network, $\aaa_{\text{in}} = \inn{\N}$ and 
$\aaa_{\text{out}} = \out{\N}$. A typing $\N:T$
is \emph{valid} iff it is sound:
\begin{description}
\item[(soundness)] 
      Every $f_0 : \aaa_{\text{in}}\cup\aaa_{\text{out}}\to\nreals$ satisfying 
      $T$ can be extended to a feasible flow $f \in\fullSem{\N}$.
\end{description}
\Hide{
The stronger notion of ``principal typing'' is an
\emph{exact} characterization of the network's input-output behavior,
in contrast to the generally \emph{approximate} characterization
provided by a ``valid typing''.}
We say the typing $\N:T$ for $\N$ is a \emph{principal typing}
if it is both sound \emph{and} complete: 
\begin{description}
\item[(completeness)] Every feasible flow $f\in\fullSem{\N}$ satisfies $T$.
\end{description}
More succintly,  using the IO-semantics $\ioSem{\N}$ instead of the full
semantics $\fullSem{\N}$, the typing $\N:T$ is \emph{valid} iff
$\poly{T}\subseteq \ioSem{\N}$, and it is \emph{principal} iff
$\poly{T} = \ioSem{\N}$.


\Hide{
If there are no feasible flows in $\N$, then the empty typing
$T=\varnothing$, \ie, the typing that assigns the empty interval to
every $A\in\power{\aaa_{\text{in}}\cup\aaa_{\text{out}}}$,
is a principal (and only valid) typing for $\N$. No feasible
flow satisfies $\varnothing$. In this case, $\poly{T}=\varnothing$.
}

\Hide
{
If $\N_1 : T_1$ and $\N_2 : T_2$ are typings for networks $\N_1$ and
$\N_2$ with similar input and output dimensions, we write $T_1\equiv
T_2$ whenever $\poly{T_1} \approx \poly{T_2}$ and say that $T_1$ and $T_2$ 
are \emph{equivalent}.
}

A useful notion in type theories is \emph{subtyping}.  If
$T_1$ is a \emph{subtype} of $T_2$, in symbols $T_1 <: T_2$,
this means that any object of type $T_1$ can be
safely used in a context where an object of type $T_2$ is expected: 
\begin{description}
\item[(subtyping)] 
   $\quad T_1 <: T_2\quad$ iff $\quad\poly{T_2} \subseteq \poly{T_1}$.
\end{description} 
Our subtyping relation is contravariant w.r.t. the subset 
relation, \ie, the supertype $T_2$ is more restrictive as a set of
flows than the subtype $T_1$.

\begin{proposition}[Principal Typings Are Subtypes of Valid Typings]
\label{prop:subtyping}
If $(\N : T_1)$ is a principal typing,
and $(\N : T_2)$ a valid typing for the same $\N$, then
$T_1 <: T_2$.
\end{proposition}

\Hide
{
\begin{proof}
Given an arbitrary $f :
\aaa_{\text{in}}\cup\aaa_{\text{out}}\to\nreals$, we want to show that
if $f$ satisfies $T_2$, then $f$ satisfies $T_1$, \ie, any point in
$\poly{T_2}$ is also in $\poly{T_1}$.  If $f$ satisfies $T_2$, then
$f$ can be extended to a feasible flow $f'$. Because $T_1$ is
principal, $f'$ satisfies $T_1$. This implies that the
restriction of $f'$ to $\aaa_{\text{in}}\cup\aaa_{\text{out}}$, which
is exactly $f$, satisfies $T_1$. 
\end{proof}
}

Any two principal typings $T_1$ and $T_2$ of the same network are not
necessarily identical, but they always denote the same polytope, as
formally stated in the next proposition.

\Hide{
\begin{lemma}
\label{lem:principal-and-uniformly-tight}
Let $(\N:T)$ and $(\N:T')$ be typings for the same $\N$. 
If $T$ and $T'$ are uniformly tight
and $\poly{T} = \poly{T'}$, then $T = T'$. 
\end{lemma}
}

\Hide
{
\begin{proof}
This is a straightforward consequence of
Proposition~\ref{prop:converting-to-uniformly-tight}
and its Coroallary~\ref{cor:converting-to-uniformly-tight}.
\end{proof}
}

\begin{proposition}[Principal Typings Are Equivalent]
\label{cor:equivalence-of-principal-typings}
\label{prop:equivalence-of-principal-typings}
If $(\N:T_1)$ and $(\N:T_2)$ are two principal typings for the same
network specification $\N$, then $T_1 \equiv T_2$. Moreover,
if $T_1$ and $T_2$ are tight, then $T_1 = T_2$.
\end{proposition}

\Hide
{
\begin{proof}
Both $(\N:T_1)$ and $(\N:T_2)$ are valid. Hence, by
Proposition~\ref{prop:subtyping}, both $T_1 <: T_2$ and $T_2 <: T_1$. 
This implies that $T_1\equiv T_2$. When $T_1$ and $T_2$ are uniformly
tight, then the equality $T_1 = T_2$ follows from
Lemma~\ref{lem:principal-and-uniformly-tight}.
\end{proof}
}

\Hide
{
\begin{corollary}
Let $\N$ be a network specification. Among the valid typings for
$\N$, the principal typings are all equivalent and minimal w.r.t.
to the subtyping ordering ``$<:$''.
\end{corollary}
}

\Hide
{
\begin{proof}
Immediate from Propositions~\ref{prop:subtyping} 
and~\ref{prop:equivalence-of-principal-typings}.
\end{proof}
}

\Hide
{
\section{Other Properties and Open Problems of Network Typings}
\label{sect:necessary-conditions}
\input{open-problems}
}

\section{Inferring Typings for Small Networks}
\label{sect:existence}
\label{sect:inference}

\begin{theorem}[Existence of Principal Typings]
\label{thm:principal-typings}
Let $\A$ be a small network. We can effectively compute a 
principal and uniformly tight typing $T$ for $\A$.%
\end{theorem}

\begin{example}
\label{ex:six-and-eight-node-networks}
Consider again the two small networks $\A$ and $\B$ from
Example~\ref{ex:illustrate-inductive-def}. 
We assign capacities to their
arcs and compute their respective principal typings. The sets of arcs in 
$\A$ and $\B$ are, respectively:
\(
  \aaa = \Set{a_1,\ldots,a_{11}}\text{ and }
  \bbb = \Set{b_1,\ldots,b_{16}}.
\)
All the lower-bounds and most of the upper-bounds are trivial, \ie,
they do not restrict flow. Specifically, the lower-bound capacity on
every arc is $0$, and the upper-bound capacity on every arc is a
``very large number'', unless indicated otherwise in
Figure~\ref{fig:six-and-eight-node} by the numbers in rectangular
boxes, namely:
\begin{alignat*}{5}
&U(a_5) = 5,\quad &&U(a_8) = 10,\quad &&U(a_{11}) = 15,
   && &&\text{non-trivial upper-bounds in $\A$},
\\
&U(b_5) = 3, &&U(b_6) = 2, &&U(b_{9}) = 2, &&U(b_{10}) = 10, 
      \qquad &&\text{non-trivial upper-bounds in $\B$},
\\
&U(b_{11}) = 8,\quad &&U(b_{13}) = 8,\quad &&U(b_{15}) = 10,\quad &&U(b_{16}) = 7,  
      \qquad &&\text{non-trivial upper-bounds in $\B$}.
\end{alignat*}
We compute the principal typings $T_{\A}$ of $\A$ and $T_{\B}$ of
$\B$, by assigning a bounded interval to every subset of 
$\Set{a_1,a_2,a_3,a_4}$ and $\Set{b_1,b_2,b_3,b_4}$,
respectively. This is a total of 15 intervals for each,
ignoring the empty set to which we assign the empty interval $\varnothing$.
We use the construction in the proof (omitted in this paper, included
in the full report~\cite{kfouryDSL:2011}) of Theorem~\ref{thm:principal-typings}
to compute $T_{\A}$ and $T_{\B}$. 
%
%
%
\begin{alignat*}{5}
&\text{$T_{\A}$ assignments}: \qquad
\\[1.5ex]
  & \framebox{$a_1:[0,15]$}\quad 
  && \fbox{$a_2:[0,25]$}\quad && \fbox{$-a_3:[-15,0]$}\quad 
  && \fbox{$-a_4:[-25,0]$}
\\
  & \fbox{$a_1+a_2:[0,30]$}\quad && \underline{a_1-a_3:[-10,10]}\quad 
  && a_1-a_4:[-25,15]
\\
  & a_2-a_3:[-15,25]\quad && \underline{a_2-a_4:[-10,10]}
  && \fbox{$-a_3-a_4:[-30,0]$}\quad 
\\
  & a_1+a_2-a_3: [0,25]\qquad&& a_1+a_2-a_4:[0,15] \qquad
   && a_1-a_3-a_4: [-25,0]\qquad && a_2-a_3-a_4: [-15,0]\quad
\\
  & a_1+a_2-a_3-a_4: [0,0]
\end{alignat*}
\begin{alignat*}{5}
&\text{$T_{\B}$ assignments}: \qquad
\\[1.5ex]
  & \fbox{$b_1:[0,15]$}\quad && \fbox{$b_2:[0,25]$}\quad 
  && \fbox{$-b_3:[-15,0]$}\quad && \fbox{$-b_4:[-25,0]$}
\\
  & \fbox{$b_1+b_2:[0,30]$}\quad && \underline{b_1-b_3:[-10,12]}\quad
  && b_1-b_4:[-25,15]
\\
  & b_2-b_3:[-15,25]\quad && \underline{b_2-b_4:[-12,10]}
  && \fbox{$-b_3-b_4:[-30,0]$}\quad 
\\
  & b_1+b_2-b_3: [0,25]\qquad&& b_1+b_2-b_4:[0,15] \qquad
   && b_1-b_3-b_4: [-25,0]\qquad && b_2-b_3-b_4: [-15,0]\quad
\\
  & b_1+b_2-b_3-b_4: [0,0]
\end{alignat*}
The types in rectangular boxes are those of $\inT{T_{\A}}$ and
$\inT{T_{\B}}$ which are equivalent, and those of $\outT{T_{\A}}$ and
$\outT{T_{\B}}$ which are also equivalent. Thus,
$\inT{T_{\A}}\equiv \inT{T_{\B}}$ and $\outT{T_{\A}}\equiv \outT{T_{\B}}$.  
Nevertheless, $T_{\A}\not\equiv T_{\B}$, the
difference being in the (underlined) types assigned to some 
subsets mixing input and output arcs:
\begin{itemize}
\item $[-10,10]$ assigned by $T_{\A}$ to $\Set{a_1,a_3}$\ \ $\neq$
      \ \ $[-10,12]$ assigned by $T_{\B}$ to the corresponding $\Set{b_1,b_3}$,
\item $[-10,10]$ assigned by $T_{\A}$ to $\Set{a_2,a_4}$\ \ $\neq$
      \ \ $[-12,10]$ assigned by $T_{\B}$ to the corresponding $\Set{b_2,b_4}$.
\end{itemize}
\Hide
{
It is not difficult to check that:
\begin{alignat*}{4}
& \poly{T_{\A}} &&= && \{\Angles{r_1,r_2,r_3,r_4}\in {\reals}^4\,|
    &&\, 0\leqslant r_1,r_3\leqslant 15;\ 0\leqslant r_2,r_4\leqslant 25;
\\
& && && &&-10\leqslant r_1-r_3\leqslant 10;\ -10\leqslant r_2-r_4\leqslant 10; 
    \ 0\leqslant r_1+r_2 = r_3+r_4\leqslant 30\,\}
\\
& \poly{T_{\B}} &&= && \{\Angles{r_1,r_2,r_3,r_4}\in {\reals}^4\,|
    &&\, 0\leqslant r_1,r_3\leqslant 15;\ 0\leqslant r_2,r_4\leqslant 25;
\\
& && && &&-10\leqslant r_1-r_3\leqslant 12;\ -12\leqslant r_2-r_4\leqslant 10; 
    \ 0\leqslant r_1+r_2 = r_3+r_4\leqslant 30\,\}
\end{alignat*}
}
In this example, $T_{\B} <: T_{\A}$ because $\poly{T_{\A}}\subseteq\poly{T_{\B}}$.
The converse does not hold.
\Hide
{
Among other things, this example shows that satisfaction of the inequalities
in~(\ref{four}) and~(\ref{five}) in Definition~\ref{def:type-satisfaction-bis}
does not necessarily imply satisfaction of the inequalities in~(\ref{three-bis}).
} 
As a result, there are feasible flows in $\B$
which are not feasible flows in $\A$. 
\Hide
{
For example, if we set:
\begin{alignat*}{5}
   &f_0(a_1)\ \ &&=\ \ &&f_0(b_1)\ \ &&=\ \ &&15
\\
   &f_0(a_2)\ &&=\ &&f_0(b_2)\ &&=\ &&0
\\
   &f_0(a_3)\ &&=\ &&f_0(b_3)\ &&=\ &&3
\\
   &f_0(a_4)\ &&=\ &&f_0(b_4)\ &&=\ &&12
\end{alignat*}
it is easy to define a feasible flow $f$ extending $f_0$ in $\B$ 
but not in $\A$.
}
\end{example}

\begin{figure}[!ht] 
\begin{center}
\hspace*{.2in}
\begin{minipage}[b]{0.45\linewidth}
\includegraphics[scale=.3,trim=0cm 12.50cm 0cm 0.5cm,clip]{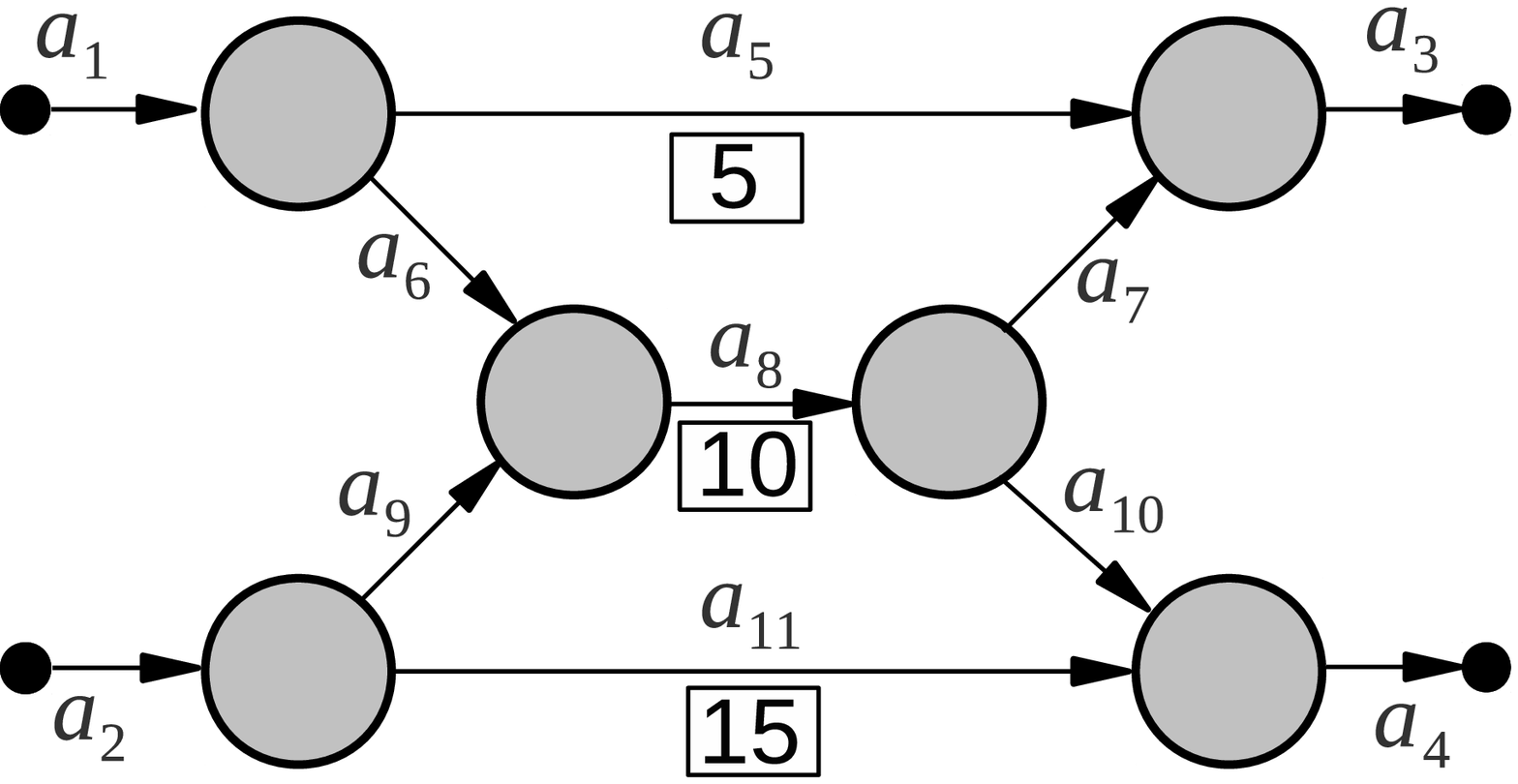}
\end{minipage}
\begin{minipage}[b]{0.45\linewidth} 
\includegraphics[scale=.3,trim=0cm 10.50cm 0cm 0.5cm,clip]{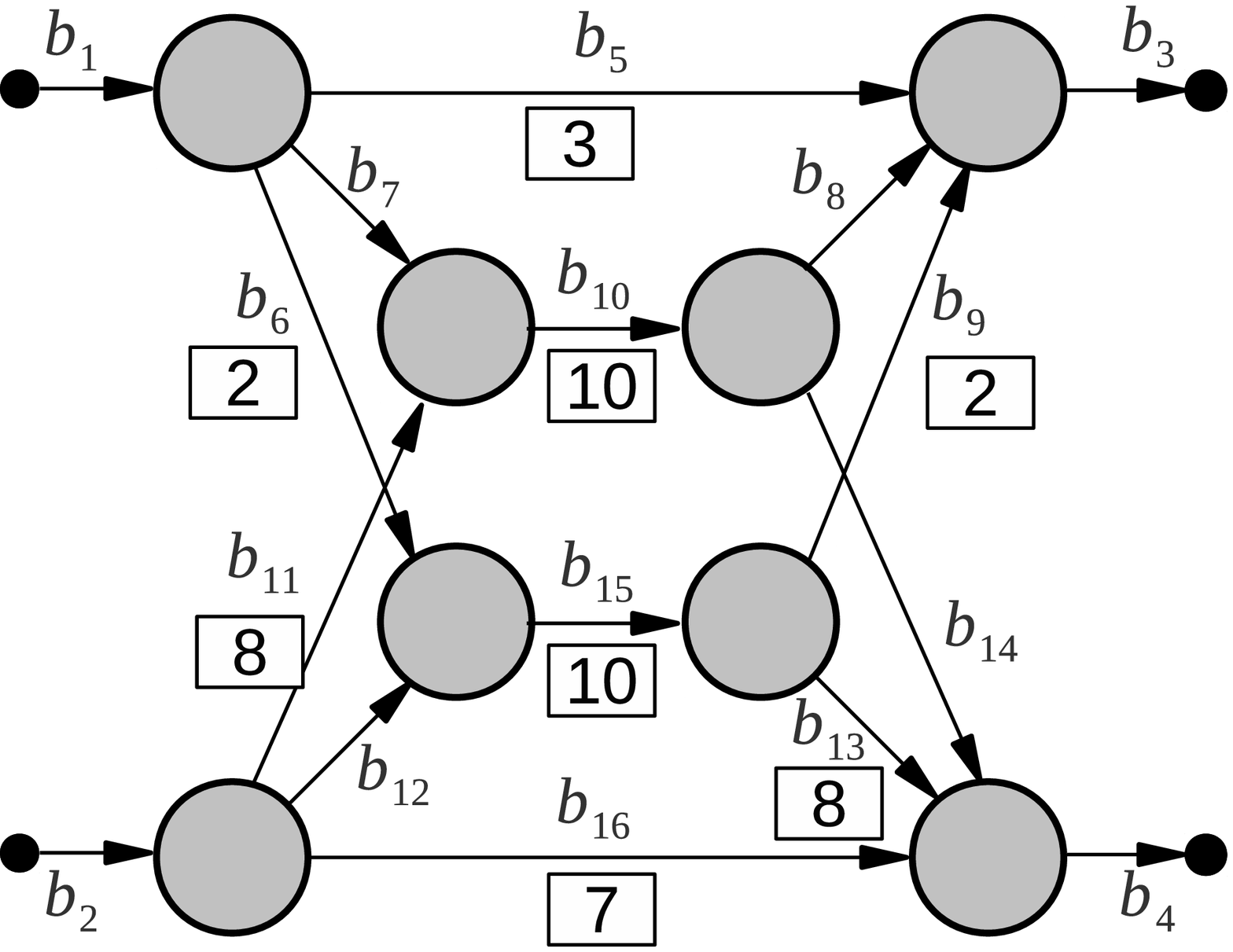}
\end{minipage}
\caption{An assignment of arc capacities
         for small networks $\A$ (on the left) and $\B$ (on the right)
         in Example~\ref{ex:six-and-eight-node-networks}.}
\label{fig:six-and-eight-node}
\end{center}
\end{figure}

\section{A Typing System}
\label{sect:inference-bis}
\label{sect:typing-rules}

We set up a formal system for assigning typings to network
specifications. The
process of inferring typings, based on this system, is deferred to
Section~\ref{sect:typing-inference}.  We need several preliminary
definitions.

\subsection{Operations on Typings}
\label{sect:operations}

Let $(\N_1:T_1)$ and $(\N_2:T_2)$ be two typings for 
two networks $\N_1$ and $\N_2$. 
The four arc sets: $\inn{\N_1}$, $\out{\N_1}$, $\inn{\N_2}$, and
$\out{\N_2}$, are pairwise disjoint. By our inductive
definition in Section~\ref{sect:inductive},
$\inn{\N_1}\cup\inn{\N_2}$ is the set of input arcs,
and $\out{\N_1}\cup\out{\N_2}$ the set of output arcs,
for the network specification $\bigl(\ConnP{\N_1}{\N_2}\bigr)$.
We define the typing $\bigl(\ConnPT{T_1}{T_2}\bigr)$ for
the specification $\bigl(\ConnP{\N_1}{\N_2}\bigr)$ as follows:
\[
   \bigl(\ConnPT{T_1}{T_2}\bigr)(A) = 
         \begin{cases}
         T_1(A) & \text{if $A\subseteq \inn{\N_1}\cup\out{\N_1}$},
         \\[2ex]
         T_2(A) & \text{if $A\subseteq \inn{\N_2}\cup\out{\N_2}$},
         \\[2ex]
         T_1(A_1) \oplus T_2(A_2)
             & \text{if $A = A_1\cup A_2$ where}
         \\
             & \text{$A_1\subseteq \inn{\N_1}\cup\out{\N_1}$
                     and $A_2\subseteq \inn{\N_2}\cup\out{\N_2}$}.
         \end{cases}
\]
where the operation ``$\oplus$'' on intervals is defined 
as follows: $[r_1,r_2]\oplus [r_1',r_2'] = [r_1+r_1',r_2+r'_2]$.
 
\begin{lemma}
\label{lem:typing-of-parallel}
If $(\N_1:T_1)$ and $(\N_2:T_2)$ are principal typings, respectively
valid typings, then so is the typing 
$\bigl((\ConnP{\N_1}{\N_2}):(\ConnPT{T_1}{T_2})\bigr)$ principal, 
respectively valid.
\end{lemma}

\Hide
{
\begin{proof} Delayed.
\end{proof}
}

Let $(\N:T)$ be a typing with $\Angles{a,b}\in\out{\N}\times
\inn{\N}$, with
$\dimI{\N}=\Angles{a_1,\ldots,a_{\ell}}$ and 
$\dimO{\N}=\Angles{a_{\ell+1},\ldots,a_{m}}$, so that $b = a_i$
and $a = a_j$ for some $1\leqslant i\leqslant\ell$ and
$\ell+1\leqslant j\leqslant m$.
In the full report~\cite{kfouryDSL:2011}
we explain how to define a typing we denote $\loopT{T}{\Angles{a,b}}$ from
the given typing $T$ for the network specification $\Loop{\Angles{a,b}}{\N}$ 
satisfying the equation:
\(
    \poly{\loopT{T}{\Angles{a,b}}}
    \ =\ 
    \poly{T} \cap \poly{a=b}
\)
where
\[
   \poly{a=b}\ =\ \Set{\,\Angles{r_1,\ldots,r_m}\in {\reals}^m\;|
        \;r_i = r_j\,}\quad
        \text{where\ $b = a_i$ and $a = a_j$ 
              with $1\leqslant i\leqslant\ell< j\leqslant m$}.
\]

\begin{lemma}
\label{lem:typing-of-loop}
If $(\N:T)$ is a principal (respectively, valid) typing 
and $\Angles{a,b}\in\inn{\N}\times\out{\N}$,  
then \emph{$\bigl(\Loop{\Angles{a,b}}{\N}:\loopT{T}{\Angles{a,b}}\bigr)$}
is a principal (respectively, valid) typing.
\end{lemma}

\medskip

\subsection{Typing Rules}
\label{sect:rules}

The system is in Figure~\ref{fig:typing-rules}, where we follow
standard conventions in formulating the rules. We call
$\Env$ a \emph{typing environment}, which
is a finite set of \emph{typing assumptions} for holes, each of
the form $(X:T)$. If $(X:T)$ is a typing assumption, with 
$\inn{X} = \aaa_{\text{in}}$ and $\out{X} = \aaa_{\text{out}}$, then
$T:\power{\aaa_{\text{in}}\cup\aaa_{\text{out}}}\to\reals\times\reals$.

\Hide
{
In the rule \textsf{\sc Let}, assumptions are discharged from the
context $\Env$. This is not essential, because we assume there is at
most one binding occurrence for every hole in a network specification
and we purposely avoid any process of ``reducing'' a network
specification whereby all let-bindings have been eliminated. (Review
the conditions for well-formedness in
Section~\ref{sect:well-formedness} to back up these comments). We
discharge assumptions in the rule \textsf{\sc Let} for conciseness and
only to indicate which holes in a network specification remain
unbound.
}

If a typing $T$ is derived for a network specification $\N$ 
according to the rules in Figure~\ref{fig:typing-rules}, 
it will be the result of deriving an \emph{assertion} 
(or \emph{judgment}) of the form ``$\Judgement{\Env}{\N}{T}$''.
If $\N$ is closed, then this final typing judgment will 
be of the form ``$\Judgement{}{\N}{T}$'' where all typing
assumptions have been discharged.

{
\begin{figure}[!ht] 
\noindent      
\begin{minipage}{1.0\textwidth}
\small
{ 
\begin{tabular}{lllll}
& \textsf{\sc Hole}
  \ &$\dfrac{\ (X : T) \in\ \Env\ }
          {\ \Judgement{\Env}{\rename{X}{i}}{\rename{T}{i\;}}\ }\qquad$ 
    & $i\geqslant 1$ is the smallest available renaming index
\\[3ex]
&\textsf{\sc Small}
  $\quad$ &$\dfrac{\ }
             {\ \Judgement{\Env}{\A}{T}\ } \qquad$
           & $T$ is a typing for small network $\A$
\\[3ex]
&\textsf{\sc Par} $\quad$
  & $\dfrac{\ \Judgement{\Env}{\N_1}{T_1}\qquad\Judgement{\Env}{\N_2}{T_2}\ }
          {\ \Judgement{\Env}{(\ConnP{\N_1}{\N_2})}{(\ConnPT{T_1}{T_2})}\ }\qquad$
\\[3ex]
&\textsf{\sc Bind} $\quad$
  & $\dfrac{\ \Judgement{\Env}{\N}{T}\ }
    {\ \Judgement{\Env}{\Loop{\Angles{a,b}}{\N}}{\loopT{T}{\Angles{a,b}}}\ }
    \qquad$
  & $\Angles{a,b} \in \out{\N}\times \inn{\N}$
\\[3ex]
&\textsf{\sc Let} $\quad$
   & $\dfrac{\ \Judgement{\Env}{\M}{T_1} \qquad
     \ \Judgement{\Env\cup\Set{(X:T_2)}}{\N}{T}\ }
     {\ \Judgement{\Env}{(\;\Let{X}{=\M}{\N}\;)}{T}\ }$ &
     $T_1\approx T_2$
\\
\end{tabular}
}
\end{minipage}
\caption{Typing Rules for Flow Networks.}\small\smallskip
        The operations $(\ConnPT{T_1}{T_2})$ and
        $\loopT{T}{\Angles{a,b}}$ are defined 
        in Section~\ref{sect:operations}. A derivation according 
        to the rules is stopped from the moment a judgment
        $\Judgement{\Env}{\N}{T}$ is reached such that $\poly{T} = \varnothing$,
        at which point $\N$ is rejected as ``unsafe''.
\label{fig:typing-rules} 
\end{figure} 
}

\begin{theorem}[Existence of Principal Typings]
\label{thm:principal-typings-in-general}
Let $\N$ be a closed network specification and $T$ a typing for $\N$
derived according to the rules in Figure~\ref{fig:typing-rules}, \ie,
the judgment ``$\Judgement{}{\N}{T}$'' is derivable according to the
rules. 
If the typing of every small 
network $\A$ in $\N$ is principal (resp., valid)
for $\A$, then $T$ is a principal (resp., valid) 
typing for $\N$.
\end{theorem}

\section{Inferring Typings for Flow Networks in General}
\label{sect:typing-inference}

The main difficulty in typing inference is in relation to 
$\textbf{\textsf{let}}$-bindings. 
Consider a specification $\N$ of the form $(\Let{X}{= \M}{\PP})$. 
Let $\aaa_{\text{in}} = \inn{X}$ and $\aaa_{\text{out}} = \out{X}$.
Suppose $X$ occurs $n\geqslant 1$ times in $\PP$, so that 
its input/output arcs are renamed in each of the $n$ occurrences 
according to: 
\(
   \rename{(\aaa_{\text{in}}\cup\aaa_{\text{out}})}{1}\ ,
   \ \ldots\ ,
   \ \rename{(\aaa_{\text{in}}\cup\aaa_{\text{out}})}{n}.
\)
A typing for $X$ and for its occurrences $\rename{X}{i}$ in $\PP$
can be given \emph{concretely} or \emph{symbolically}. If concretely,
then these typings are functions of the form:
\[
    T_{X}:
    \power{\aaa_{\text{in}}\cup\aaa_{\text{out}}}\to\reals\times\reals
    \quad\text{and}\quad
    \rename{T_{X}}{i}:
    \power{\rename{\aaa_{\text{in}}}{i}
    \cup\rename{\aaa_{\text{out}}}{i}}\to\reals\times\reals
\]
for every $1\leqslant i\leqslant n$. According to the typing rule
\textsc{Hole} in Figure~\ref{fig:typing-rules}, a valid typing for $\N$ requires that:
\(
    T_X \approx \rename{T_{X}}{1}\approx\cdots\approx\rename{T_{X}}{n}.
\)
If symbolically, then for every $B\subseteq \aaa_{\text{in}}\cup\aaa_{\text{out}}$, 
the interval $T_{X}(B)$ is written as $[x_B,y_B]$
where the two ends $x_B$ and $y_B$ are yet to be determined,
and similarly for $\rename{T_{X}}{i}(B)$ and every
$B\subseteq \rename{\aaa_{\text{in}}}{i}\cup\rename{\aaa_{\text{out}}}{i}$.
%
We can infer a typing for $\N$ in one of two ways, which 
produce the same end result but whose organizations are very different:
\begin{description}
\item[(sequential)] First infer a principal typing $T_{\M}$ for $\M$,
  then use $k$ copies $\rename{T_{\M}}{1},\ldots,\rename{T_{\M}}{n}$
  to infer a principal typing $T_{\PP}$ for $\PP$, which is also a
  principal typing $T_{\N}$ for $\N$.
\item[(parallel)] Infer principal typings $T_{\M}$ for $\M$
  and $T_{\PP}$ for $\PP$, separately. 
  $T_{\PP}$ is parametrized by the typings $\rename{T_{X}}{i}$
  written symbolically. A typing for $\N$ is obtained by
  setting lower-end and upper-end parameters in $\rename{T_{X}}{i}$
  to corresponding lower-end and upper-end values in
  $T_{\M}$.
\end{description}
Both approaches are \emph{modular}, in that both are syntax-directed
according to the inductive definition of $\N$. However, the parallel
approach has the advantage of being independent of the order in which
the inference proceeds (\ie, it does not matter whether  $T_{\M}$
is inferred before or after, or simultaneously with, $T_{\PP}$).
We therefore qualify the parallel approach as being
additionally \emph{fully compositional}, in contrast to the sequential
approach which is not. Moreover, the latter requires that the whole
specification $\N$ be known before typing inference can start,
justifying the additional qualification of being a
\emph{whole-specification} analysis.
The sequential approach is simpler to define
and is presented in full in~\cite{kfouryDSL:2011}. 
We delay the examination of the parallel/fully-compositional approach
to a follow-up report.

\section{Semantics of Flow Networks Relative to Objective Functions}
\label{sect:relativized-semantics}

Let $\N$ be a network, with 
$\aaa_{\text{in}} = \inn{\N}$, $\aaa_{\text{out}} = \out{\N}$, and
$\aaa_{\text{\#}} = \inter{\N}$.
We write $\aaa_{\text{out,\#}}$ to denote $\aaa_{\text{out}}\uplus\aaa_{\text{\#}}$,
the set of all arcs in $\N$ excluding the input arcs.
An \emph{objective function} selects a subset of feasible flows
that minimize (or maximize) some quantity. 
We list two possible objective functions, among
several others, commonly considered in ``traffic
engineering'' (see \cite{BalonLeduc:icon2006} for example).
\begin{description}
\item[\emph{Minimize Hop Routing} (HR)] A minimum hop route is a route
     with minimal number of links. \\ 
     Given a feasible flow $f\in\fullSem{\N}$, we define the quantity $\hr{f}
     = \sum_{a\in \aaa_{\text{out,\#}}} f(a)$. Given two feasible flows
     $f_1,f_2\in\fullSem{\N}$, we write $\LE{\symhr}{f_1}{f_2}$ iff two
     conditions: 
     \begin{itemize}
     \item 
     $\rest{f_1}{\aaa_{\text{in}}} = \rest{f_2}{\aaa_{\text{in}}}$,\ \ and
     \item 
     $\hr{f_1} < \hr{f_2}$.
     \end{itemize}
     Note that we compare $f_1$ and $f_2$ using $\LE{\symhr}{}{}$ only
     if they assign the same values to the input arcs, which implies
     in particular that $f_1$ and $f_2$ carry equal flows across
     $\N$. It can be shown that $\hr{f_1} < \hr{f_2}$ holds iff $f_1$
     is non-zero on fewer arcs in $\aaa_{\text{out,\#}}$ than $f_2$, \ie,
     \[ 
        \size{\Set{\,a\in \aaa_{\text{out,\#}}\;|\;f_1(a)\neq 0\,}}
        \ <
      \ \size{\Set{\,a\in \aaa_{\text{out,\#}}\;|\;f_2(a)\neq 0\,}}
     \]
     We write $\LEQ{\symhr}{f_1}{f_2}$ to mean $\LE{\symhr}{f_1}{f_2}$ 
     or $\hr{f_1} = \hr{f_2}$.
\item[\emph{Minimize Arc Utilization} (AU)] The utilization
    of an arc $a$ is defined as $u(a) = f(a)/U(a)$. \\
    Given a feasible flow $f\in\fullSem{\N}$, we define the quantity $\au{f}
    = \sum_{a\in \aaa_{\text{out,\#}}} u(a)$. Given two feasible flows 
    $f_1,f_2\in\fullSem{\N}$, we write $\LE{\symau}{f_1}{f_2}$ iff two conditions:
     \begin{itemize}
     \item 
     $\rest{f_1}{\aaa_{\text{in}}} = \rest{f_2}{\aaa_{\text{in}}}$,\ \ and
     \item 
     $\au{f_1} < \au{f_2}$.
     \end{itemize}
    It can be shown that $\au{f_1} < \au{f_2}$ holds iff:
    \[ 
       \sum \Set{\,1/U(a)\;|\; a\in \aaa_{\text{out,\#}}
        \text{ and } f_1(a)\neq 0\,}
        \ <
      \ \sum \Set{\,1/U(a)\;|\; a\in \aaa_{\text{out,\#}}
        \text{ and } f_2(a)\neq 0\,}
     \]
     Minimizing arc utilization corresponds to computing
     ``shortest paths'' from inputs to outputs using 
     $1/U(a)$ as the metric on every arc in $\aaa_{\text{out,\#}}$.
     We write $\LEQ{\symau}{f_1}{f_2}$ to mean $\LE{\symau}{f_1}{f_2}$ 
     or $\au{f_1} = \au{f_2}$.
\Hide{
\item[\emph{Minimize Mean Delay} (MD)] The mean delay
    of an arc $a$ can be measured by $d(a) = 1/(U(a)-f(a))$. \\
    Given a feasible flow $f\in\fullSem{\N}$, we define the quantity $\md{f}
    = \sum_{a\in \aaa_{\text{out,\#}}} d(a)$. Given two feasible flows 
    $f_1,f_2\in\fullSem{\N}$, we write $\LE{\symmd}{f_1}{f_2}$ iff two conditions:
     \begin{itemize}
     \item 
     $\rest{f_1}{\aaa_{\text{in}}} = \rest{f_2}{\aaa_{\text{in}}}$, \ \ and
     \item 
     $\md{f_1} < \md{f_2}$.
     \end{itemize}
    It can be shown that $\md{f_1} < \md{f_2}$ holds iff:
    \[ 
       \sum \Set{\,1/(U(a)-f_1(a))^2\;|\; a\in \aaa_{\text{out,\#}} \,}
          \ <
      \ \sum \Set{\,1/(U(a)-f_2(a))^2\;|\; a\in \aaa_{\text{out,\#}} \,}
    \]
     In contrast to \textbf{HR} and \textbf{AU}, the minimization of
     \textbf{MD} depends on the flow carried by the arcs in 
     $\aaa_{\text{out,\#}}$.    
     We write $\LEQ{\symmd}{f_1}{f_2}$ to mean $\LE{\symmd}{f_1}{f_2}$ 
     or $\md{f_1} = \md{f_2}$.
}
\end{description} 

\Hide{
\begin{remark}
The definition of the functions $\symhr$, $\symau$, and $\symmd$, is a
summation over $\aaa_{\text{out,\#}} = \aaa_{\text{out}}\cup\aaa_{\text{\#}}$.
An alternative is to sum over $\aaa_{\text{in}}\cup\aaa_{\text{\#}}$ or over 
only $\aaa_{\text{\#}}$. There are minor technical differences between these
three alternatives. Our choice for summing over $\aaa_{\text{out,\#}}$ simplifies 
a little clause 5, in the formal semantics below.
\end{remark}
}

\noindent
For the rest of this section, consider a fixed objective
$\ooo\in\Set{\symhr,\symau,\ldots}$.  We relativize the formal
semantics of flow networks as presented in
Section~\ref{sect:flows}. To be correct, our relativized semantics
requires that the objective $\ooo$ be an ``additive aggregate
function''.

\begin{definition}{Additive Aggregate Functions} 
Let $\N$ be a network and consider its set $\fullSem{\N}$ of feasible flows.
A function $\alpha:\fullSem{\N}\to\nreals$ is an \emph{additive aggregate}
if $\alpha(f)$ is of the form $\sum_{a\in \aaa_{\text{out,\#}}} \theta(f,a)$
for some function $\theta: \fullSem{\N}\times\aaa_{\text{out,\#}}\to\nreals$. 
\end{definition}

The particular objective functions $\symhr$ and $\symau$ considered above
are additive aggregate. For $\symhr$, the corresponding function $\theta$
is the simplest and defined by $\theta(f,a) = f(a)$. And for $\symau$, 
the corresponding function $\theta$ is defined by $\theta(f,a) = f(a)/U(a)$.
All the objective functions considered in~\cite{BalonLeduc:icon2006}
are additive aggregate.

The \emph{full semantics of a flow network $\N$
relative to objective $\ooo$}, denoted $\fullSem{\N\,|\,\ooo}$, will be
a set of triples each of the form $\Angles{f,B,r}$ where:
\begin{itemize}
\item $f\in\fullSem{\N}$, \ie, $f$ is a feasible flow in $\N$,
\item $B\subseteq \inn{\N}\cup\out{\N}$,
\item $r = \ooo(f)$,
\end{itemize}
such that, for every feasible flow $g\in\fullSem{\N}$, if $\rest{f}{B}
= \rest{g}{B}$ then $\ooo(g)\geqslant r$.  The information
provided by the parameters $B$ and $r$ allows us to
determine $\fullSem{\N\,|\,\ooo}$ compositionally, \ie,
in clause 5 in the definition of $\fullSem{\N\,|\,\ooo}$ below: We can define
the semantics of a network $\M$ relative to $\ooo$ from the semantics
of its \emph{immediate} constituent parts relative to $\ooo$. Informally, if 
$\Angles{f,B,r}\in\fullSem{\N\,|\,\ooo}$, then among all feasible flows that
agree on $B$, flow $f$ minimizes $\ooo(f)$. We include the parameter
$r = \ooo(f)$ in the triple to avoid re-computing $\ooo$
from scratch at every step of the induction, 
by having to sum over \emph{all} the arcs of $\N$.
Based on the preceding, starting with small networks $\A$, 
we define the full semantics of $\A$ relative to the objective $\ooo$
as follows:
\begin{alignat*}{4}
  &\fullSem{\A\,|\,\ooo}\ &&=\ &&\bigl\{\,
   \Angles{f,B,r}\;\bigl|\; &&f\in\fullSem{\A},\ B\subseteq\inn{\A}\cup\out{\A},
    \ r = \ooo(f),
\\ 
  & && && &&\text{and for every $g\in\fullSem{\A}$, if $\rest{f}{B} = \rest{g}{B}$
            then $\ooo(f)\leqslant \ooo(g)$}\,\bigr\}
\end{alignat*}
The IO-semantics $\ioSem{\A\,|\,\ooo}$ of the small network
$\A$ relative to the objective $\ooo$ is:
\begin{alignat*}{4}
  &\ioSem{\A\,|\,\ooo}\ &&=\ &&\SET{\,
       \Angles{\rest{f}{A},B,r}\;\bigl|\; 
       &&\Angles{f,B,r}\in\fullSem{\A\,|\,\ooo}\,}
\end{alignat*}
where $A = \inn{\A}\cup\out{\A}$.
As in Section~\ref{sect:flows}, the full semantics 
$\fullSem{X\,|\,\ooo}$ and the IO-semantics $\ioSem{X\,|\,\ooo}$ of
a hole $X$ relative to the objective $\ooo$ are the same.
Let $\aaa_{\text{in}}=\inn{X}$ and $\aaa_{\text{out}}=\out{X}$, so that:
\begin{alignat*}{2}
   &\fullSem{X\,|\,\ooo}\ =\ \ioSem{X\,|\,\ooo}\ \subseteq
   \ \bigl\{ \Angles{f,B,s}\;\bigl|
   \; &&f:\aaa_{\text{in}}\cup\aaa_{\text{out}}\to\nreals,
   \ B\subseteq\aaa_{\text{in}}\cup\aaa_{\text{out}},
   \ s\in\nreals,\text{ and $f$ is bounded}\bigr\}
\end{alignat*}
Again, as in Section~\ref{sect:flows}, 
$\fullSem{X\,|\,\ooo} = \ioSem{X\,|\,\ooo}$ is not uniquely
defined.  Whether this assigned semantics of $X$ will work
depends on whether the condition in clause 4 below is satisfied. 

\medskip
We define $\fullSem{\M\,|\,\ooo}$ for every subexpression $\M$ of
$\N$, by induction on the structure of the specification $\N$. 
The five clauses here are identical to those in Section~\ref{sect:flows}, 
except for the $\ooo$-relativization. The only non-trivial clause is the 5th and
last; Proposition~\ref{prop:relativized-full-semantics} establishes
the correctness of this definition:
\begin{enumerate}
\item If $\M = \A$, then 
      $\fullSem{\M\,|\,\ooo} = \fullSem{\A\,|\,\ooo}$.
\item If $\M = \rename{X}{i}$, then $\fullSem{\M\,|\,\ooo} =
      \rename{\fullSem{X\,|\,\ooo}}{i\,}$.
\item If $\M = \bigl(\ConnP{\PP_1}{\PP_2}\bigr)$,
      then 
   \begin{alignat*}{3}
      &\fullSem{\M\,|\,\ooo}\ &&= 
      \ &&\SET{\,\Angles{\ConnPT{f_1}{f_2}, B_1\cup B_2, r_1+r_2}\;\bigl|
      \;\Angles{f_1,B_1,r_1}\in\fullSem{\PP_1\,|\,\ooo}\text{ and }
      \Angles{f_2,B_2,r_2}\in\fullSem{\PP_2\,|\,\ooo}\,}
   \end{alignat*}
\item If $\M = \bigl(\Let{X}{=\PP}{{\PP}'}\bigr)$,
      then $\fullSem{\M\,|\,\ooo} = \fullSem{{\PP}'\,|\,\ooo}$, 
      provided two conditions:%
      \footnote{Review footnote~\ref{crucial-foot} for
      the meaning of ``$\approx$''.} 
      \begin{enumerate}
      \item
          $\dimIO{X}\ \approx\ \dimIO{\PP}$,
      \item 
          $\fullSem{X\,|\,\ooo} \ \approx
           \Bigl\{\,\Angles{\rest{g}{A},C,r}\;\bigl|
           \;\Angles{g,C,r}\in\fullSem{\PP\,|\,\ooo}\,\Bigr\}$
\Hide{          for every $f:\inn{X}\cup\out{X}\to\nreals$, 
          $B\subseteq \inn{X}\cup\out{X}$, and $r\in\nreals$,
            \[
               \Angles{f,B,r}\in\fullSem{X\,|\,\ooo}\quad \text{iff}
               \quad\text{there is }
               \Angles{g,C,r}\in\fullSem{\PP\,|\,\ooo}
               \text{ such that $f \approx \rest{g}{A}$ and $B\approx C$, }
            \]}
          where $A = \inn{\PP}\cup\out{\PP}$.
      \end{enumerate}
\item If $\M = \Loop{\Angles{a,b}}{\PP}$, then 
      \begin{alignat*}{3}
      &\fullSem{\M\,|\,\ooo}\ &&=  
        \ \Bigl\{\,\Angles{f,B,r}\;\bigl|\;\;\;
        && 
        \Angles{f,B\cup\Set{a,b},r}\in\fullSem{\PP\,|\,\ooo},
        \ f(a) = f(b), 
\\
  & && &&
     \text{and for every $\Angles{g,B\cup\Set{a,b},s}\in\fullSem{\PP\,|\,\ooo}$}
\\
  & && &&\text{if $g(a) = g(b)$ and $\rest{f}{B} = \rest{g}{B}$ 
               then $r\leqslant s$}\,\Bigr\}
      \end{alignat*}
\end{enumerate}
We define $\ioSem{\N\,|\,\ooo}$ from $\fullSem{\N\,|\,\ooo}$:
\(
  \ioSem{\N\,|\,\ooo} =\ \SET{\,\Angles{\rest{f}{A},B,r}\,\bigl|\,
      \Angles{f,B,r}\in\fullSem{\N\,|\,\ooo}\,}
\)
where $A = \inn{\N}\cup\out{\N}$.

\begin{proposition}[Correctness of Flow-Network Semantics, Relativized]
\label{prop:relativized-full-semantics}
Let $\N$ be a network specification and let 
$\ooo$ be an additive aggregate objective. For
every $f:\aaa_{\text{in}}\cup\aaa_{\text{out}}\cup\aaa_{\text{\#}}\to\nreals$,
every $B\subseteq \aaa_{\text{in}}\cup\aaa_{\text{out}}$, and every
$r\in\nreals$, it is the case that:
\begin{alignat*}{2}
&\Angles{f,B,r}\in\fullSem{\N\,|\,\ooo}
  \quad\text{iff}\quad &&f\in\fullSem{\N} \text{ and } r=\ooo(f) \text{ and}
\\
& &&\text{for every $g\in\fullSem{\N}$, if 
          $\rest{f}{B} = \rest{g}{B}$ then $\ooo(g)\geqslant r$.}
\end{alignat*}
In words, for every $B\subseteq \aaa_{\text{in}}\cup\aaa_{\text{out}}$,
among all feasible flows in $\N$ that agree on $B$, we include in  
$\fullSem{\N\,|\,\ooo}$ those that are $\ooo$-optimal and exclude
from $\fullSem{\N\,|\,\ooo}$ those that are not.
\end{proposition}

\section{A Relativized Typing System}
\label{sect:relativized-typing-system}

Let $\ooo$ be an additive aggregate objective, \eg, one of those
mentioned in Section~\ref{sect:relativized-semantics}. Assume $\ooo$
is fixed and the same throughout this section.
\Hide{, except in
Example~\ref{ex:illustrating-objective-functions} where we instantiate
$\ooo$ to the objectives $\symhr$, $\symau$, and $\symmd$.}
Let $\N$ be a closed network specification. According to 
Section~\ref{sect:typing-rules}, if the judgment ``$\Judgement{}{\N}{T}$'' 
is derivable using the rules in Figure~\ref{fig:typing-rules} and $T$ is 
a valid typing, then $\poly{T}$ is a set of feasible IO-flows in $\N$,
\ie, $\poly{T} \subseteq \ioSem{\N}$. And if $T$ is principal, then
in fact $\poly{T} = \ioSem{\N}$.

In this section, judgments are of the form
``$\Judgement{}{\N}{(T,\objA{}{})}$'' and derived using the rules in
Figure~\ref{fig:relativized-typing-rules}.  We call $(T,\objA{}{})$ a
\emph{relativized typing}, where $T$ is a typing as before and
$\objA{}{}$ is an auxiliary function depending on the objective
$\ooo$.  If $T$ is a valid (resp. principal) typing for $\N$, then
once more $\poly{T} \subseteq \ioSem{\N}$ (resp. $\poly{T} =
\ioSem{\N}$), but now the auxiliary $\objA{}{}$ is used to select
members of $\poly{T}$ that minimize $\ooo$.

If this is going to work at all, $\objA{}{}$ should not 
inspect the whole of $\N$.  Instead, $\objA{}{}$ should be
defined inductively from the relativized typings for only the immediate
constituent parts of $\N$.
We first explain what the auxiliary $\objA{}{}$
tries to achieve, and then explain how it can be defined inductively.
The objective $\ooo$ is already defined
on $\fullSem{\N}$, as in Section~\ref{sect:relativized-semantics}.
We now define it on $\ioSem{\N}$. For every $f\in\ioSem{\N}$, let:
\begin{alignat*}{3} 
  & \ooo(f)\ =&&\ \min\;\SET{\,\ooo(f')\;\bigl|
    \;f'\in\fullSem{\N}\text{ and $f'$ extends $f$}\,}.
\end{alignat*} 
As before, let $\aaa_{\text{in}} = \inn{\N}$ and $\aaa_{\text{out}} =
\out{\N}$.  Let $T$ be a valid typing for $\N$, so that
$\poly{T}\subseteq\ioSem{\N}$.  For economy of writing, let
$\F=\poly{T}$. Relative to this $T$, we define the function
$\objA{T}{}$ as follows:

\bigskip
\fbox{
\hspace*{-.3in}
\begin{minipage}{.8\textwidth}
\vspace*{-.2in}
\begin{alignat*}{3}
 &\objA{T}{} &&:&&
    \power{\aaa_{\text{in}}\cup\aaa_{\text{out}}}\to \power{\F\times\nreals}
\\
 &\objA{T}{}(B) 
  \ && =\ &&\SET{\,\Angles{f,r}\;\bigl|\;f\in\F,\ r=\ooo(f),\ \text{and }
  \text{for every $g\in\F$, if $\rest{f}{B} = \rest{g}{B}$,
        then $r\leqslant \ooo(g)$\,}}
\end{alignat*}
\end{minipage}
}

\bigskip
\noindent
where $B \in \power{\aaa_{\text{in}}\cup\aaa_{\text{out}}}$.
In words, $\objA{T}{}(B)$ selects $f$ provided, among all members of
$\F\subseteq\ioSem{\N}$ that agree with $f$ on $B$, $f$ is $\ooo$-optimal --
and also appends to $f$ its $\ooo$-value $r$ for book-keeping purposes.
Whenever the context makes it clear, we omit the subscript ``$T$''
from ``$\objA{T}{}$'' and simply write ``$\objA{}{}$''.

The trick here is to define the auxiliary function
$\objA{}{}$ for $\N$ from the corresponding auxiliary
functions for the immediate constituent parts of $\N$. The only non-trivial
step follows the 5th and last clause in the definition of 
$\fullSem{\N\,|\,\ooo}$ in Section~\ref{sect:relativized-semantics}.

\begin{definition}{Valid and Principal Relativized Typings}
\label{defn:valid-and-principal}
Let $(T,\objA{}{})$ be a relativized typing for $\N$, where
$\inn{\N} = \aaa_{\text{in}}$ and $\out{\N} = \aaa_{\text{out}}$. We
define $\Poly{T,\objA{}{}}$ as a set of triples:
\[
  \Poly{T,\objA{}{}}\ =\ \Set{\,\Angles{f,B,r}\;|\;
  B\subseteq\aaa_{\text{in}}\cup\aaa_{\text{out}}\text{ and }
  \Angles{f,r}\in\objA{}{}(B)\,}
\]
We call this function ``$\Poly{}$'' because of its close association
with ``$\poly{}$'', as it is easy to see that:
\begin{alignat*}{2}
  &\Poly{T,\objA{}{}}\ =\ 
  \bigl\{\,\Angles{f,B,r}\;\bigl|\;&&f\in\poly{T},
   \ B\subseteq\aaa_{\text{in}}\cup\aaa_{\text{out}},\ r=\ooo(f),
\\ 
  & &&\text{and for all $g\in\poly{T}$ 
     if $\rest{f}{B} = \rest{g}{B}$ then $\ooo(f)\leqslant\ooo(g)$}\,\bigl\}
\end{alignat*}
We say the relativized typing $\bigl(\N:(T,\objA{}{})\bigr)$ 
is \emph{valid} iff $\Poly{T,\objA{}{}}\subseteq\ioSem{\N\,|\,\ooo}$,
and we say it is \emph{principal} iff $\Poly{T,\objA{}{}}=\ioSem{\N\,|\,\ooo}$.
\end{definition}

A case of particular interest 
is when $B = \aaa_{\text{in}}$. Suppose 
$\Angles{f,\aaa_{\text{in}},r}\in\Poly{T,\objA{}{}}$. This means that,
among all feasible flows $g$ in $\N$ agreeing with
$f$ on $\aaa_{\text{in}}$, $f$ is $\ooo$-optimal with $\ooo(f) = r$.
\Hide
{
More, in fact, for every $\ooo$-optimal feasible flow $g$ in $\N$ agreeing 
with $f$ on $\aaa_{\text{in}}$, we will have $\ooo(g) = r$, so that:
\[
  \SET{\,\Angles{g,\aaa_{\text{in}},r}\;\bigl|\;\rest{f}{\aaa_{\text{in}}} =
  \rest{g}{\aaa_{\text{in}}}
  \text{ and $g$ is $\ooo$-optimal feasible flow in $\N$}\,}
  \ \subseteq\  \Poly{T,\objA{}{}}
\]
Note $\Angles{g,\aaa_{\text{in}},r}$ provides no information
about the path taken by $\ooo$-optimal feasible flow $g$
\emph{inside} $\N$. It only says the $\ooo$-value of $g$ is
$r$, assuming $g$ indeed uses an $\ooo$-optimal path through $\N$.
Nor does it say anything about the values assigned to output arcs
by an $\ooo$-optimal feasible flow $g$ through $\N$, although
we only know that 
\( 
  \sum f(\aaa_{\text{in}}) = \sum f(\aaa_{\text{out}}) =
  \sum g(\aaa_{\text{in}}) = \sum g(\aaa_{\text{out}})
\) 
because of flow conservation.
}

\subsection{Operations on Relativized Typings}
\label{sect:operations-relativized}

There are two different operations on relativized typings
depending on how they are obtained from previously
defined relativized typings. These two operations are
``$\ConnPT{(T_1,\objA{1}{})}{(T_2,\objA{2}{})}$'' and
``$\loopT{(T,\objA{}{})}{\Angles{a,b}}$'', whose definitions are based
on clauses 3 and 5 in the inductive definition of
$\fullSem{\N\,|\,\ooo}$ in Section~\ref{sect:relativized-semantics}.

Let $\bigl(\N_1:(T_1,\objA{1}{})\bigr)$ and 
$\bigl(\N_2:(T_2,\objA{2}{})\bigr)$ be two relativized typings for 
two networks $\N_1$ and $\N_2$. Recall that the
the four arc sets: $\inn{\N_1}$, $\out{\N_1}$, $\inn{\N_2}$, and
$\out{\N_2}$, are pairwise disjoint. We define the relativized typing
$(T,\objA{}{}) = \ConnPT{(T_1,\objA{1}{})}{(T_2,\objA{2}{})}$ for 
the specification $\bigl(\ConnP{\N_1}{\N_2}\bigr)$ as follows:
\begin{itemize}
\item $T = (\ConnPT{T_1}{T_2})$, as defined at the beginning of 
      Section~\ref{sect:operations},
\item for every $B_1\subseteq \inn{\N_1}\cup\out{\N_1}$
      and every $B_2\subseteq \inn{\N_2}\cup\out{\N_2}$:
\[
     \objA{}{}(B_1\cup B_2)\ =
     \ \SET{\,\Angles{(\ConnPT{f_1}{f_2}),r_1+r_2}\;\bigl|
     \;\Angles{f_1,r_1}\in\objA{1}{}(B_1)\text{ and }
     \Angles{f_2,r_2}\in\objA{2}{}(B_2)\,}
\]
\end{itemize}

\begin{lemma}
\label{lem:relativized-typing-of-parallel}
If the relativized typings
$\bigl(\N_1:(T_1,\objA{1}{})\bigr)$ and 
$\bigl(\N_2:(T_2,\objA{2}{})\bigr)$ are principal, resp. valid,
then so is the relativized typing
$\bigl(\ConnP{\N_1}{\N_2}\bigr):
\bigl(\ConnPT{(T_1,\objA{1}{})}{(T_2,\objA{2}{})}\bigr)$ principal, resp.
valid.
\end{lemma}

Let $\bigl(\PP:(T,\objA{}{})\bigr)$ be a relativized typing for network
specification $\PP$.  We define the relativized typing $(T^*,\objA{}{*}) =
\loopT{(T,\objA{}{})}{\Angles{a,b}}$ for the network 
$\Loop{\Angles{a,b}}{\PP}$ as follows:
\begin{itemize}
\item $T^* = \loopT{T}{\Angles{a,b}}$,
      as defined in Section~\ref{sect:operations},
\item for every $B\subseteq (\inn{\PP}\cup\out{\PP})-\Set{a,b}$:
\begin{alignat*}{2}
     &\objA{}{*}(B)\ =
     \ \bigl\{\,\Angles{\rest{f}{B},r}\;\bigl|
     \; &&\Angles{f,r}\in\objA{}{}(B\cup\Set{a,b}),\; f(a)=f(b),
     \text{ and for all $\Angles{g,s}\in\objA{}{}(B\cup\Set{a,b})$}
\\ 
     & &&\text{if $g(a) = g(b)$ and $\rest{f}{B}=\rest{g}{B}$ 
         then $r\leqslant s$}\,\bigr\}
\end{alignat*}
\end{itemize}

\begin{lemma}
\label{lem:relativized-typing-of-bind}
If the relativized typing $\bigl(\PP:(T,\objA{}{})\bigr)$
is principal, resp. valid, then so is the relativized typing
\emph{$\bigl(\Loop{\Angles{a,b}}{\PP}:\loopT{(T,\objA{}{})}{\Angles{a,b}}\bigr)$} 
principal, resp. valid.
\end{lemma}

\subsection{Relativized Typing Rules}
\label{sect:rules-relativized}

\begin{figure}[!ht] 
\noindent      
\begin{minipage}{1.0\textwidth}
\small
{ 
\begin{tabular}{lllll}
& \textsf{\sc Hole}
  &$\dfrac{\ (X : (T,\objA{}{})) \in\ \Env\ }
          {\ \Judgement{\Env}{\rename{X}{i}}
    {(\rename{T}{i\;},\rename{\objA{}{}}{i\;})}\ }$ %
  & $i\geqslant 1$ is smallest 
\\[-1ex] & & & 
  available renaming index
\\[2ex]
&\textsf{\sc Small} 
  &$\dfrac{\ }
          {\ \Judgement{\Env}{\A}{(T,\objA{}{})}\ }$
   \qquad & $(T,\objA{}{})$ is a relativized 
\\[-1ex] & & & 
   typing for small network $\A$
\\[2ex]
&\textsf{\sc Par} 
  & $\dfrac{\ \Judgement{\Env}{\N_1}{(T_1,\objA{1}{})}
    \qquad\Judgement{\Env}{\N_2}{(T_2,\objA{2}{})}\ }
   {\ \Judgement{\Env}{(\ConnP{\N_1}{\N_2})}
   {\ConnPT{(T_1,\objA{1}{})}{(T_2,\objA{2}{})}}\ }
   $ 
\\[3ex]
&\textsf{\sc Bind} 
  & $\dfrac{\ \Judgement{\Env}{\N}{(T,\objA{}{})}\ }
    {\ \Judgement{\Env}{\Loop{\Angles{a,b}}{\N}}
       {\loopT{(T,\objA{}{})}{\Angles{a,b}}}\ }
    $ 
  & $\Angles{a,b} \in \out{\N}\times \inn{\N}$
\\[3ex]
&\textsf{\sc Let} 
   & $\dfrac{\ \Judgement{\Env}{\M}{(T_1,\objA{1}{})}
            \qquad
 \ \Judgement{\Env\cup\Set{X:(T_2,\objA{2}{})}}{\N}
   {(T,\objA{}{})}\  }
 {\ \Judgement{\Env}
     {(\;\Let{X}{=\M}{\N})}
     {(T,\objA{}{})}\ }$ 
   & $(T_1,\objA{1}{})\approx (T_2,\objA{2}{})$
\end{tabular}
}
\end{minipage}
\caption{Relativized Typing Rules for Flow Networks.}\small\smallskip
        The operations 
        ``$\ConnPT{(T_1,\objA{1}{})}{(T_2,\objA{2}{})}$''
        and ``$\loopT{(T,\objA{}{})}{\Angles{a,b}}$''
        are defined in Section~\ref{sect:operations-relativized}.
        A derivation according 
        to the rules is stopped from the moment a judgment
        $\Judgement{\Env}{\N}{(T,\objA{}{})}$ is reached such that 
        $\Poly{T,\objA{}{}} = \varnothing$, at which point
        $\N$ is rejected as ``unsafe''.
\label{fig:relativized-typing-rules} 
\end{figure} 

\begin{theorem}[Existence of Relativized Principal Typings]
\label{thm:principal-relativized-typings-in-general}
Let $\N$ be a closed network specification and $(T,\objA{}{})$ a
relativized typing for $\N$ derived according to the rules in
Figure~\ref{fig:relativized-typing-rules}, \ie, the judgment
``$\Judgement{}{\N}{(T,\objA{}{})}$'' is derivable according to the
rules.
If the relativized typing of every small network $\A$ in $\N$ is 
principal (resp., valid) for $\A$, then $(T,\objA{}{})$ 
is a principal (resp., valid) relativized typing for $\N$.
\end{theorem}

\vspace{-.2in}

\section{Related and Future Work}
\label{sect:related}

Ours is not the only study that uses \emph{intervals} as types and
\emph{polytopes} as typings. There were earlier attempts that heavily
drew on linear algebra and polytope theory, mostly initiated
by researchers who devised ``types as abstract interpretations'' --
see~\cite{Cousot97-1} and references therein. However, the motivations
for these earlier attempts were entirely different and applied to 
programming languages unrelated to our DSL. For example, polytopes
were used to define ``invariant safety properties'', or
``types'' by another name, for \textsc{Esterel} -- an imperative
synchronous language for the development of reactive systems
\cite{polys}.

Apart from the difference in motivation with earlier works,
there are also technical differences in the use of polytopes.  Whereas
earlier works consider polytopes defined by unrestricted linear
constraints \cite{cousot-halbwachs-78,polys}, our polytopes are
defined by linear constraints where every coefficient is $+1$ or $-1$,
as implied by our Definitions~\ref{def:flow-conservation},
\ref{def:capacity-constraints}, \ref{def:feasible-flows},
and~\ref{def:type-satisfaction}. Ours are identical to the
linear constraints (but not necessarily the linear objective function)
that arise in the \emph{network simplex method} \cite{Cunningham79},
\ie, linear programming applied to problems of network flows. There is
still on-going research to improve network-simplex algorithms
(\eg, \cite{rashidi-tsang-09}), which will
undoubtedly have a bearing on the efficiency of typing inference for
our DSL.

Our polytopes-cum-typings are far more restricted than polytopes
in general. Those of particular interest to us correspond to \emph{valid}
typings and \emph{principal} typings. As of now, we do not have a characterization 
-- algebraic or even syntactic on the shape of linear constraints -- of 
polytopes that are \emph{valid network typings} 
(or the more restrictive \emph{principal network typings}). 
Such a characterization will likely guide and improve 
the process of typing inference.

Let $\N$ be a network specification, with $\aaa_{\text{in}} = \inn{\N}$ and
$\aaa_{\text{out}} = \out{\N}$. Another source of current inefficiency is
that valid and principal typings for $\N$ tend to be ``over-specified'', 
as they unnecessarily assign an interval-cum-type to \emph{every} subset of 
${\aaa_{\text{in}}\uplus\aaa_{\text{out}}}$. 
Several examples in~\cite{kfouryDSL:2011} illustrate this kind of
inefficiency. This will lead
us to study \emph{partial typings} 
\(
   T : \power{\aaa_{\text{in}}\uplus\aaa_{\text{out}}}
       \rightharpoonup\reals\times\reals
\),
which assign intervals to some, not necessarily all, subsets of 
${\aaa_{\text{in}}\uplus\aaa_{\text{out}}}$.
Such a partial mapping $T$ can always be extended to a total mapping 
$T' : \power{\aaa_{\text{in}}\uplus\aaa_{\text{out}}}\to\reals\times\reals$,
in which case we write $T\subseteq T'$. We say 
the partial typing $T$ is \emph{valid} for $\N$ if \emph{every} (total)
typing $T'\supseteq T$ is valid for $\N$, and we say $T$
is \emph{minimal valid} for $\N$ if $T$ is valid for $\N$ \emph{and}
for every partial typing $T''$ for $\N$ such that $T''\subsetneq T$, 
\ie, $T''$ assigns strictly fewer intervals than $T$, it is the case that 
$T\not\equiv T'$. And similarly for the definitions of partial typings that 
are \emph{principal} and \emph{minimal principal} for $\N$.  

As alluded in the Introduction and again in Remark~\ref{rem:invariance-of-semantics}, 
we omitted an operational semantics of
our DSL in this paper to stay clear of complexity issues arising from
the associated rewrite (or reduction) rules. Among other benefits,
relying on a denotational semantics allowed us to harness this
complexity by performing a static analysis, via our typing theory,
without carrying out a naive hole-expansion
(or \textbf{\textsf{let-in}} elimination).  We thus traded the
intuitively simpler but costlier operational semantics for the more
compact denotational semantics.

However, as we introduce other more complex constructs involving holes in follow-up
reports (\textbf{\textsf{try-in}}, \textbf{\textsf{mix-in}}, and 
\textbf{\textsf{letrec-in}} mentioned in the Introduction and in
Remark~\ref{rem:other-derived-constructors} of
Section~\ref{sect:inductive}) this trade-off will diminish in
importance. An operational semantics of our DSL involving these 
more complex hole-binders will bring it closer in line with various calculi
involving \emph{patterns} (similar to our \emph{holes} in many ways, different
in others) and where rewriting consists in eliminating \emph{pattern-binders}. See
\cite{baldan,barthe:pure,Cirstea2004,Cirstea03rewritingcalculus,jay-kesner:esop06} 
and references therein. It remains
to be seen how much of the theory developed for these pattern calculi
can be adapted to an operational semantics of our DSL.

\bibliographystyle{eptcs}
\bibliography{generic,extra}

\end{document}